\documentclass[acmsmall,screen,nonacm]{acmart}

\usepackage{main}

\begin{document}

\title{Potential Functions as Types}
\subtitle{A Synthetic Modal Formulation of Amortized Cost}

\begin{abstract}
  Amortized analysis can be framed from the physicist's view, amenable to manual verification in dependent type theory using potential functions, and the banker's view, amenable to automated inference in substructural type theory using type-level credit annotations.
  In this work, we synthesize these perspectives in Calf, a dependent type theory for cost verification.
  From the physicist's view, we present a fracture and gluing theorem that renders every type as containing a fusion of an abstraction function and a potential function.
  By construction, every program between two such types must preserve abstraction, to facilitate modularity of behavior, and conserve potential, to facilitate modularity of cost.
  Incorporating the banker's view, we synthetically construct type operators for credits and debits.
  We then define Giralf, a graded substructural dependent type theory for programming with credits and debits, which is semantically interpreted as a sub-language of Calf.
  Finally, we adapt an inference algorithm to transform a limited class of Calf programs into Giralf counterparts, automating the cost analysis of common algorithms in Calf.
\end{abstract}
 
\author{Harrison Grodin}
\orcid{0000-0002-0947-3520}
\email{hgrodin@cs.cmu.edu}

\author{Ethan Chu}
\orcid{0009-0005-6041-0313}
\email{ethanchu@cs.cmu.edu}

\author{Runming Li}
\orcid{0000-0001-7600-9069}
\email{runmingl@cs.cmu.edu}

\affiliation{
	\institution{Carnegie Mellon University}
	\department{Computer Science Department}
	\streetaddress{5000 Forbes Ave.}
	\city{Pittsburgh}
	\state{PA}
	\postcode{15213}
	\country{USA}
}

\author{Jan Hoffmann}
\orcid{0000-0001-8326-0788}
\email{jhoffmann@cmu.edu}

\author{Robert Harper}
\orcid{0000-0002-9400-2941}
\email{rwh@cs.cmu.edu}

\affiliation{
	\institution{Carnegie Mellon University}
	\department{Computer Science Department}
	\streetaddress{5000 Forbes Ave.}
	\city{Pittsburgh}
	\state{PA}
	\postcode{15213}
	\country{USA}
}

\authorsaddresses{}


\maketitle

\section{Introduction} \label{sec:introduction}

Amortized analysis, pioneered by \citeauthor{sleator-tarjan>1985-paging} [\citeyear{tarjan>1985}], is a technique for analyzing the cost of a sequence of operations on an ephemeral data structure.
Since its inception, there have been two compatible perspectives of the method---the \emph{physicist's view} and the \emph{banker's view}.

In the physicist's view, a \emph{potential function} $\Phi : \X \to \Cost$ assigns potential (\ie, future cost) to each data structure of a type $\X$, where $\Cost$ is a type representing cost (commonly the natural numbers).
Then, for an operation $f : \X \to \X$ with a \emph{true cost} $c_\top : \X \to \Cost$ and an imagined \emph{amortized cost} $c_{\ABS} : \X \to \Cost$, one proves a principle tantamount to the conservation of energy:
\begin{equation}
  c_\top(x) + \Phi(f(x)) \le \Phi(x) + c_{\ABS}(x).  \label{eq:amortization}
\end{equation}
Iterating this inequality (traditionally via a telescoping sum) ensures that the true cost of a sequence of operations is upper-bounded by the sum of the amortized costs and the initial potential.
Because it requires a proof of \cref{eq:amortization}, which could rely on arbitrarily complex facts and invariants of the data, the physicist's view is well-suited for \emph{manual verification} in dependent type theory~\citep{niu-sterling-grodin-harper>2022,grodin-harper>2024} and higher-order logic~\citep{nipkow-brinkop>2019}.

In the banker's view, cost is viewed as a coin-like resource---called a \emph{credit}---that can be saved within a data structure.
Credits can be spent later to offset the cost of an expensive operation; if all true costs are offset by credits, the amortized cost of a sequence of operations is simply the number of credits stored within the input data.
Due to their status as a resource, credits must be treated \emph{substructurally}: although credits may be wasted, they may not be duplicated.
In many common algorithms and data structures, it is possible to attach the requisite credits to a data structure automatically, placing a credit in exactly the locations where cost will later be incurred.
Thus, the banker's view is well-suited for substructural logics and type
theories~\citep{atkey>2011,mevel-jourdan-pottier>2019} as well as
\emph{automated inference}~\citep{hofmann-jost>2003,hoffmann-jost>2022}.

From either perspective, amortization is fundamentally about \emph{modularity}.
Amortized analysis does not affect the \emph{implementation} details or the true cost of data structure operations.
Instead, it exports a reasonable cost model as a \emph{cost interface} that allows client programs to reason about the cost of operations while encapsulating exactly when costs occur.
Both potential functions and credits are \emph{ghost data}, serving only to mediate between the \emph{private} reality of an implementation and the \emph{public} fictitious amortized costs presented in an interface.

In this work, we develop a technique
that unifies the physicist's method and the banker's method in a dependent type theory for modular verification of amortized cost.
    The physicist's potential functions and associated inequalities are first-class.
    Every type comes, implicitly or explicitly, equipped with a private potential function; and every program includes, implicitly or explicitly, a proof of the conservation of potential.
    A synthetic phase distinction ensures modularity, providing a stable mathematical model and amortized cost bound against which client programs may be verified.
    Within this type theory, we define the banker's credits as a type operator.
    Then, we construct a substructural sub-language for writing programs with credits and provide an inference algorithm that emits amortized cost upper bounds and corresponding certificates of soundness.

\begin{wrapstuff}[r,type=figure,width=0.35\textwidth]
  \centering
  \begin{tikzpicture}[scale=0.7]
    \begin{scope}[fill opacity=0.4]
      \fill[red]   (90:1.3cm)  circle (2.2cm);
      \fill[green] (210:1.3cm) circle (2.2cm);
      \fill[blue]  (330:1.3cm) circle (2.2cm);
    \end{scope}

    \draw (90:1.3cm)  circle (2.2cm);
    \draw (210:1.3cm) circle (2.2cm);
    \draw (330:1.3cm) circle (2.2cm);

    \node at (90:2.2cm)  {amortized cost};
    \node[align=center, text width=2cm] at (210:2.42cm) {functions \\ in types};
    \node[align=center, text width=2cm] at (330:2.4cm) {lax \\ squares};

    \node at (150:1.6cm) {AARA};
    \node at (30:1.55cm) {Calf};
    \node at (270:1.5cm) {AFAT};

    \node at (0,0) {PFAT};
  \end{tikzpicture}
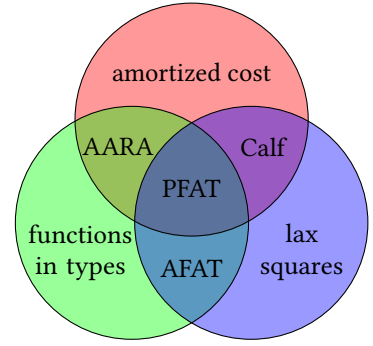
  \caption{The central ideas of this work.}\label{fig:venn}
\end{wrapstuff}

Our approach achieves these goals by synthesizing three main ideas, seamlessly integrating manual verification of amortized cost, modularity via abstraction, and automated cost inference within dependent type theory, depicted in \cref{fig:venn}.
\par
\begin{enumerate}
  \item
    We work in Calf~\citep{niu-sterling-grodin-harper>2022,grodin-niu-sterling-harper>2024}, a dependent type theory for cost verification.
    Within Calf, the conservation of energy principle used in the physicist's view of amortized analysis can be packaged as a lax commutative square~\citep{grodin-harper>2024}.
  \item
    We make use of the insights of \citet{grodin-li-harper>2026} who achieve modularity in (univalent) dependent type theory by rendering \emph{abstraction functions as types}~(AFAT).
    Seen via a modal fracture and gluing theorem~\citep{rijke-shulman-spitters>2020}, every type contains an abstraction function, and every function between types contains a commutative square ensuring abstraction is preserved.
    To accommodate the cost effect of Calf, the authors permit a weaker notion of lax commutativity on costs.
  \item
    We incorporate ideas from the Automatic Amortized Resource Analysis (AARA) family of substructural type theories~\citep{hofmann-jost>2003,hoffmann-jost>2022}, which includes types equipped with credits to represent the banker's view of amortized analysis.
    Semantically, the types of AARA are interpreted as containing both a set of values and a potential function, and the soundness theorem ensures that potential is conserved.
    Because credits are spent locally, AARA supports automated inference of cost for common classes of programs, using linear programming to ensure sufficient credits are always available.
\end{enumerate}
Unifying these ideas, we render \emph{potential functions as types}~(PFAT).
In the following, we provide additional background on each of these ideas, which we then make use of throughout the work.

\subsection{Calf: Cost Analysis in Dependent Type Theory}\label{sec:introduction:calf}

This work takes place in the Calf type theory~\citep{niu-sterling-grodin-harper>2022}, which extends dependent type theory with an adjoint layer supporting a notion of cost.
Calf is a dependent variation of call-by-push-value~\citep{levy>2003,ahman-ghani-plotkin>2016,vakar>thesis,pedrot-tabareau>2019}.
As such, it includes two sorts of types, the \emph{value types} and the \emph{computation types}:
\begin{align*}
  \text{Val.}~\X,\Y,\Z &\coloncolonequals \U{\A} \mid \UnitV \mid \ProdV{\X}{\Y} \mid \ZeroV \mid \SumV{\X}{\Y} \mid \NatV \mid \ListV{\X} \mid {\textstyle\DSumV{x : \X}{\Y(x)}} \mid \VV \mid \CC \mid \cdots  \\
  \text{Comp.}~\A,\B,\C &\coloncolonequals \F{\X} \mid \UnitC \mid \ProdC{\A}{\B} \mid \PowerC{\X}{\A} \mid \DPowerC{x : \X}{\A(x)} \mid \cdots
\end{align*}
The universe of value types $\VV$ and the universe of computation types $\CC$ are both, themselves, value types~\citep{krishnaswami-pradic-benton>2015}.%

Let $(\Cost, \le_{\Cost}, 0, +)$ be an ordered commutative monoid representing cost, typically chosen to be the natural numbers $(\NatV, \le_{\NatV} 0, +)$.
The \emph{cost effect} is a printing-like effect, available at all computation types $\A$: the effect $\charge[\A]{c}[e]$ increases the cost of computation $e : \A$ by $c : \Cost$ units of cost.
\begin{center}
  \begin{minipage}{0.45\textwidth}
    \[
      \infer
        {\Gamma \vdash c : \Cost \\ \Gamma \vdash e : \A}
        {\Gamma \vdash \charge[\A]{c}[e] : \A}
    \]
  \end{minipage}%
  \begin{minipage}{0.55\textwidth}
    \[\begin{aligned}
      \charge{0}[e] &= e \\
      \charge{c_1 + c_2}[e] &= \charge{c_1}[\charge{c_2}[e]]
    \end{aligned}\]
  \end{minipage}%
\end{center}
Every type $\A$ is equipped with a cost preorder $\le_{\A}$ on $\U{\A}$,%
\footnote{Formally, \citeauthor{grodin-niu-sterling-harper>2024} use a \emph{synthetic} notion of preorder. In univalent type theory, we inherit all types as value types, placing the requirement for a type to be a synthetic preorder on computation types.}
where $e \le_{\A} e'$ means that $e$ and $e'$ have the same behavior, although the cost of $e$ may be lower than that of $e'$~\citep{grodin-niu-sterling-harper>2024}.

Within Calf, \citet{niu-sterling-grodin-harper>2022} verify amortized costs using the physicist's method by proving the requisite conservation principles.
Later, in a refinement of Calf, \citet{grodin-harper>2024} showed that amortized analysis can be viewed as a lax commutative square using additional type constructors of the enriched effect calculus (EEC)~\citep{egger-mogelberg-simpson>2009,egger-mogelberg-simpson>2014} and linear/non-linear type theory (LNL)~\citep{benton>1995,krishnaswami-pradic-benton>2015}%
\footnote{The $\TensorC{\A}{\B}$ and $\LolliC{\A}{\B}$ types of LNL make use the commutativity of the cost monoid.}%
, including pure functions $\ExpV{\X}{\Y}$, homomorphisms $\LolliV{\A}{\B}$, sums $\ZeroC$ and $\SumC{\A}{\B}$, and copowers $\CopowerC{\X}{\A}$.
\begin{align*}
  \text{Val.}~\X,\Y,\Z &\coloncolonequals \cdots \mid \ExpV{\X}{\Y} \mid \DProdV{x : \X}{\Y(x)} \mid \LolliV{\A}{\B} \\
  \text{Comp.}~\A,\B,\C &\coloncolonequals \cdots \mid \underbrace{\ZeroC \mid \SumC{\A}{\B} \mid \CopowerC{\X}{\A} \mid \CopowerC{(x : \X)}{\A(x)}}_\text{EEC} \mid \underbrace{\TopC \mid \TensorC{\A}{\B} \mid \LolliC{\A}{\B}}_\text{LNL}
\end{align*}
Using homomorphisms $\LolliV{\A}{\B}$, \citet{grodin-harper>2024} describe a generalization of amortized analysis centered on computation types.
For an implementation type $\A[\top]$, let $g_\top : \LolliV{\A[\top]}{\A[\top]}$ be an operation on the data structure annotated with a realistic cost model; and for a specification type $\A[\ABS]$, let $g_{\ABS} : \LolliV{\A[\ABS]}{\A[\ABS]}$ be an analogous operation instead annotated with a purported amortized cost.
These two programs can be connected by defining a homomorphism $\alpha : \LolliV{\A[\top]}{\A[\ABS]}$ such that the inequality $\alpha \circ g_\top \le g_{\ABS} \circ \alpha$ holds, depicted as the following lax commutative square:
\[\begin{tikzcd}
	{\A[\top]} & {\A[\top]} \\
	{\A[\ABS]} & {\A[\ABS]}
	\arrow[mmv, "{g_\top}", from=1-1, to=1-2]
	\arrow[mmv, ""{name=0, anchor=center, inner sep=0}, "\alpha"', from=1-1, to=2-1]
	\arrow[mmv, ""{name=1, anchor=center, inner sep=0}, "\alpha", from=1-2, to=2-2]
	\arrow[mmv, "{g_{\ABS}}"', from=2-1, to=2-2]
	\arrow["\ge"{description}, draw=none, from=1, to=0]
\end{tikzcd}\]
In the case where $\A[\top] = \F{\X[\top]}$ and $\A[\ABS] = \F{\X[\ABS]}$, the maps above are of the following form:
\begin{align*}
  g_\top~(\ret{x_\top}) &\isdef \charge{c_\top(x_\top)}[\ret[f_\top(x_\top)]] &c_\top &: \ExpV{\X[\top]}{\Cost} &f_\top &: \ExpV{\X[\top]}{\X[\top]} \\
  g_{\ABS}~(\ret{x_{\ABS}}) &\isdef \charge{c_{\ABS}(x_{\ABS})}[\ret[f_{\ABS}(x_{\ABS})]] &c_{\ABS} &: \ExpV{\X[\ABS]}{\Cost} &f_{\ABS} &: \ExpV{\X[\ABS]}{\X[\ABS]} \\
  \alpha~(\ret{x_\top}) &\isdef \charge{\Phi(x_\top)}[\ret[\chi(x_\top)]] &\Phi &: \ExpV{\X[\top]}{\Cost} &\chi &: \ExpV{\X[\top]}{\X[\ABS]}
\end{align*}
\citeauthor{grodin-harper>2024} observe that the cost aspect of the inequality $\alpha \circ g_\top \le g_{\ABS} \circ \alpha$, in this case, is precisely the conservation of potential condition of \cref{eq:amortization}.
For this reason, they refer to the lax commutative square as a ``generalized amortization condition'' and to $\alpha$ as a ``behavior-relevant generalization of potential functions''.
In this work, inspired by the induced equation $\chi \circ f_\top = f_{\ABS} \circ \chi$, we take a dual perspective: we treat $\alpha$ as a cost-aware generalization of an \emph{abstraction function}.

\subsection{Abstraction Functions as Types}

To support modular abstract data types within dependent type theory, we build
on recent work on incorporating
\emph{abstraction functions}~\citep{hoare>1972} into dependent type theory~\citep{grodin-li-harper>2026}.
The authors propose an \emph{abstract phase}, a proposition $\AbsV$, to isolate abstract/public/interface-level data from concrete/private/implementation-level data using modalities in homotopy type theory~\citep{sterling-harper>2021,rijke-shulman-spitters>2020}.
This phase gives rise to a \emph{fracture and gluing} theorem guaranteeing that every type $\X$ contains precisely a concrete type, an abstract type, and an abstraction function between the two, all accessible via modal constructions defined using $\AbsV$.

When a function $\chi : \ExpV{\X[\top]}{\X[\ABS]}$ is assembled into a type $\X$, functions on $\X$ must respect abstraction.
For example, a function $f : \ExpV{\X}{\X}$ can be built out of a function $f_\top : \ExpV{\X[\top]}{\X[\top]}$ on the concrete representation, a function $f_{\ABS} : \ExpV{\X[\ABS]}{\X[\ABS]}$ on the abstract representation, and a proof of coherence up to $\chi$ (below, left).
This technique facilitates modularity by ensuring that, in addition to a concrete implementation, every program also contains a stable abstract specification on which clients can depend.
Client code is then verified in the abstract phase (\ie, assuming $\AbsV$), where both the concrete implementation and the abstraction function are erased: $\X = \X[\ABS]$ and $f = f_{\ABS}$.

\begin{center}
  \begin{minipage}{0.5\textwidth}
\[\begin{tikzcd}
	{\X[\top]} & {\X[\top]} \\
	{\X[\ABS]} & {\X[\ABS]}
	\arrow[arrv, "{f_\top}", from=1-1, to=1-2]
	\arrow[arrv, ""{name=0, anchor=center, inner sep=0}, "\chi"', from=1-1, to=2-1]
	\arrow[arrv, ""{name=1, anchor=center, inner sep=0}, "\chi", from=1-2, to=2-2]
	\arrow[arrv, "{f_{\ABS}}"', from=2-1, to=2-2]
	\arrow["{=}"{description}, draw=none, from=1, to=0]
\end{tikzcd}\]
  \end{minipage}%
  \begin{minipage}{0.5\textwidth}
\[\begin{tikzcd}
	{\X[\top]} & {\U[\F{\X[\top]}]} \\
	{\X[\ABS]} & {\U[\F{\X[\ABS]}]}
	\arrow[arrv, "{f_\top}", from=1-1, to=1-2]
	\arrow[arrv, ""{name=0, anchor=center, inner sep=0}, "\chi"', from=1-1, to=2-1]
	\arrow[arrv, ""{name=1, anchor=center, inner sep=0}, "{\U[\F{\chi}]}", from=1-2, to=2-2]
	\arrow[arrv, "{f_{\ABS}}"', from=2-1, to=2-2]
	\arrow["{\ge}"{description}, draw=none, from=1, to=0]
\end{tikzcd}\]
  \end{minipage}%
\end{center}

\citet{grodin-li-harper>2026} extend this technique to accommodate the cost effect.
Let $f_\top$ and $f_{\ABS}$ be functions satisfying the depicted program inequality (above, right), meaning that they cohere behaviorally and the cost of $f_\top$ is upper-bounded by the cost of $f_{\ABS}$.
Just as the behavior of $f_{\ABS}$ is a client-facing approximation of the true behavior of $f_\top$, the cost annotation within $f_{\ABS}$ is a client-facing upper-bound on the cost of $f_\top$.
This construction does not account for amortization exactly because the abstraction function $\chi : \ExpV{\X[\top]}{\X[\ABS]}$ is pure and thus cannot incorporate potential via the cost effect.
In the present work, we extend this development to enable the assembly of a cost-aware homomorphism $\alpha : \LolliV{\A[\top]}{\A[\ABS]}$ into a type, thus accommodating both abstraction and potential.
The cost of the client-facing $f_{\ABS} : \LolliV{\A[\ABS]}{\A[\ABS]}$ is then an \emph{amortized} upper-bound on the cost of $f_{\top} : \LolliV{\A[\top]}{\A[\top]}$, facilitating modular verification with amortized costs.

\subsection{AARA: Automatic Amortized Resource Analysis}\label{sec:introduction:aara}

To automatically infer cost bounds within our type theory,
we build on ideas from AARA, a family of substructural type systems~\citep{hofmann-jost>2003}
based on the banker's view.
There are many variants of AARA that have been developed over more
than two decades~\citep{hoffmann-jost>2022}; we focus on a core
language approximately based on that of \citet{hoffmann-hofmann>2010}.

Syntactically in AARA, types include credits; for example, the type $\Credit{c}{\A}$ stores $c$ credits alongside data of type $\A$.
Semantically, connecting to the physicist's method, every type describes not only a set of values $\AARAValues{\A}$ but also a potential function $\Phi_{\A} : \ExpV{\AARAValues{\A}}{\Cost}$ that computes the potential/credits contained within a type.
For example, $\AARAValues{\Credit{c}{\A}} \isdef \AARAValues{\A}$ and $\Phi_{\Credit{c}{\A}}(a) \isdef c + \Phi_{\A}(a)$.
Syntactically, the affine treatment of credits in AARA
ensures that credits are never duplicated; semantically, this appears
in the soundness theorem for AARA as a conservation of potential
theorem for $\Phi_{\A}$.

The design of AARA is motivated by the desire to reduce automated cost
inference to a linear programming (LP) problem.
However, this limits expressivity, requiring potential functions (usually polynomials) to be selected \emph{a priori} with limited support for manual
verification~\cite{pham-niu-glover-saad-hoffmann>2025}.

\begin{samepage}
In reference to languages in the tradition of AARA, \citeauthor{niu-sterling-grodin-harper>2022} make the following remark:
\begin{myquote}{niu-sterling-grodin-harper>2022}
  \dots it is not immediately clear how one may take better advantage of the existing type-based approaches to amortized analysis in Calf.
\end{myquote}
In this work, using the idea that potential functions can be built into types as a bridge between Calf and AARA, we incorporate credits into types to make precisely this connection.
\end{samepage}

\subsection{Contributions}

In this work, we synthesize these three major ideas: potential as
the cost of an abstraction function, abstraction functions built into types, and credits as a substructural type former.
\begin{enumerate}
\item We extend Calf to natively support a modular account of
  amortized analysis within dependent type theory, reconstructing the
  physicist's method.
  Specifically, in \cref{sec:physicist} we extend the
  work of \citet{grodin-li-harper>2026} on building abstraction
  functions into types, proving a fracture and gluing theorem for the
  universe of computation types that enable cost-aware abstraction functions that emit potential to be built into types.
  In \cref{sec:modularity}, we observe that this approach facilitates a modular notion of amortized cost interface via the abstract phase.

\item In \cref{sec:banker}, we develop a standard library for credits,
  debits, and credit-carrying data structures based on the banker's
  method inside the dependent type theory.

\item
  In \cref{sec:giralf} we define Giralf, an AARA-like graded substructural type theory that streamlines programming with credits and debits.
  Then, by providing it with a semantics in Calf, we demonstrate formally the sense in which Giralf generalizes AARA.
  Moreover, we adapt the LP-based cost
  inference algorithm for AARA to Giralf, generating certificates in Calf
  guaranteeing the soundness of inferred bounds.

\end{enumerate}
The central theorems and constructions of this work are mechanized in Cubical Agda~\citep{norell>2009,vezzosi-mortberg-abel>2019}, indicated by the \AgdaFormalized{} symbol.

\section{The Physicist's View}\label{sec:physicist}

To incorporate potential functions into types, we extend the technique of \citet{grodin-li-harper>2026} that builds functions $\ExpV{\X[\top]}{\X[\ABS]}$ into value types to the level of computation types, where a homomorphism $\LolliV{\A[\top]}{\A[\ABS]}$ can emit cost to be thought of as potential.
This development takes place in a univalent type theory that is merely extended with a single proposition, the abstract phase $\AbsV$; the rest of the constructions are defined in terms of this proposition.

\subsection{Abstraction Functions as Types: Hoare's Abstraction Functions, Synthetically}

First, we briefly review the work of \citet{grodin-li-harper>2026} that makes use of the modalities of \citet{rijke-shulman-spitters>2020} to incorporate abstraction functions into types.
The key result is a \emph{fracture and gluing} theorem stating that every type contains exactly a concrete type, an abstract type, and an abstraction function between them.
We first recall the modalities that isolate concrete and abstract types.

\begin{definition}[Abstract Modality, \citealp{grodin-li-harper>2026} \AgdaFormalized{Calf.Value.Open}]\label{def:open-modality}
  The \emph{abstract modality} $\OpenV{\X} \isdef \ExpV{\AbsV}{\X}$ marks a type as client-facing,
  where its unit $\etaOpV[\X] : \ExpV{\X}{\OpenV{\X}}$ is $\vlam{x}{\vlam{\abs*}{x}}$.
  A value type $\X$ is \emph{abstract} when $\etaOpV[\X]$ is an equivalence, and write $\VV[\op]$ for the universe of abstract value types.
\end{definition}

\begin{definition}[Concrete Modality, \citealp{grodin-li-harper>2026} \AgdaFormalized{Calf.Value.Closed}]\label{def:closed-modality}
  The \emph{concrete modality} $\ClosedV{\X}$ marks a type as irrelevant to clients.
  It is defined as the following higher-inductive type (a pushout):
  \begin{center}
    \begin{minipage}{0.5\linewidth}
      \begin{center}
      \iblock{
        \mhang{\KW{data}~\ClosedV~(\X : \VV) : \VV~\KW{where}}{
          \mrow{\LABEL{\etaClV[\X]} : \ExpV{\X}{\ClosedV{\X}}}
          \mrow{\LABEL{\starClV} : \ExpV{\AbsV}{\ClosedV{\X}}}
          \mrow{\LABEL{\_} : \DProdV*{x : \X}{\DProdV{\abs : \AbsV}{\LABEL{\etaClV[\X]} x = \LABEL{\starClV}~\abs}}}
        }
      }
      \end{center}
    \end{minipage}%
    \begin{minipage}{0.5\linewidth}
      \[\begin{tikzcd} {\ProdV{\AbsV}{\X}} & \AbsV \\
        \X & \ClosedV{\X}
        \arrow[arrv, "{\LABEL{proj}_2}"', from=1-1, to=2-1]
        \arrow[arrv, "{\LABEL{proj}_1}", from=1-1, to=1-2]
        \arrow[arrv, "{\starClV}", from=1-2, to=2-2]
        \arrow[arrv, "{\etaClV}", from=2-1, to=2-2]
        \arrow["\lrcorner"{anchor=center, pos=0.125, rotate=180, color=RegalBlue}, draw=none, from=2-2, to=1-1]
      \end{tikzcd}\]
    \end{minipage}%
  \end{center}
  Say that a value type $\X$ is \emph{concrete} when $\etaClV[\X]$ is an equivalence
  (or equivalently when $\OpenV{\X} = \UnitV$)
  and write $\VV[\cl]$ for the universe of concrete value types.
\end{definition}

For a concrete type $\X[\cl] : \VV[\cl]$ and an abstract type $\X[\op] : \VV[\op]$, an \emph{abstraction function} in the style of \citet{hoare>1972} is a function $\chi_{\cl} : \ExpV{\X[\cl]}{\ClosedV{\X[\op]}}$, where the concrete modality on the output ensures that the abstraction itself is entirely concrete and hidden from the interface.
Based on these definitions, we recall the fracture and gluing theorem, which proves that every type in the universe $\VV$ carries precisely the same data as a concrete type, an abstract type, and an abstraction function.

\begin{theorem}[Fracture and Gluing, \citealp{rijke-shulman-spitters>2020} {\AgdaFormalized[fracture-and-gluing]{Calf.Value.Glue}}]\label{thm:vfracglue}
  The following equivalence holds:
  \[
    \VV = \DSumV{\X[\cl] : \VV[\cl]}{\DSumV{\X[\op] : \VV[\op]}{\ExpV{\X[\cl]}{\ClosedV{\X[\op]}}}}
  \]
\end{theorem}
\begin{proof}[Proof Sketch]
  In the forward direction, fracture a type $\X$ into the concrete type $\X[\cl] \isdef \ClosedV{\X}$, the abstract type $\X[\op] \isdef \OpenV{\X}$, and the abstraction function $\ClosedV{\etaOpV[\X]} : \ExpV{\ClosedV{\X}}{\ClosedV{\OpenV{\X}}}$.
  In the reverse direction, define $\GlueV{\X[\cl]}{\X[\op]}{\chi_\cl} \isdef \ProdV{\X[\cl]}{\X[\op]}[\ClosedV{\X[\op]}] \isdef \DSumV{(x_\cl, x_\op) : \ProdV{\X[\cl]}{\X[\op]}}{\chi_\cl(x_\cl) = \etaClV(x_\op)}$.
\end{proof}
\begin{corollary}[{\AgdaFormalized[fracture-and-gluing-square]{Calf.Value.Glue}}]\label{cor:vfracglue-map}
  Every $f : \ExpV{\X}{\Y}$ consists of a concrete function $f_\cl : \ExpV{\ClosedV{\X}}{\ClosedV{\Y}}$ and an abstract function $f_\op : \ExpV{\OpenV{\X}}{\OpenV{\Y}}$ that cohere (formally, $\ClosedV{\etaOpV[\Y]} \circ f_\cl = \ClosedV[f_\op \circ \etaOpV[\X]]$).
\end{corollary}

Using gluing, it is possible to treat any function $\chi : \ExpV{\X[\top]}{\X[\ABS]}$ as an abstraction function $\ClosedV[\etaOpV \circ \chi] : \ExpV{\ClosedV{\X[\top]}}{\ClosedV{\OpenV{\X[\ABS]}}}$ and thus build it into a type; accordingly, we sometimes refer to $\chi$ as an abstraction function.
For convenience, we define a utility that builds an arbitrary $\chi$ into a type.
\begin{definition}[{\AgdaFormalized[Abstraction]{Calf.Value.Abstraction}}]\label{def:AbstractionV}
  Define $
    \GlueeV{\X[\top]}{\X[\ABS]}{\chi} \isdef \GlueV{\ClosedV{\X[\top]}}{\OpenV{\X[\ABS]}}{\ClosedV[\etaOpV \circ \chi]}
  $.
\end{definition}
\begin{corollary}[{\AgdaFormalized[Abstraction-id]{Calf.Value.Abstraction}}]\label{cor:AbstractionV-id}
  By \cref{thm:vfracglue}, it is the case that $\X = \GlueeV{\X}{\X}{\idV[\X]}$ for all $\X$.
\end{corollary}
\begin{lemma}[{\AgdaFormalized[square']{Calf.Value.Abstraction}}]\label{lem:AbstractionV-exp}
  To define a map
  $f : \ExpV{\GlueeV{\X[\top]}{\X[\ABS]}{\chi}}{\GlueeV{\Y[\top]}{\Y[\ABS]}{\psi}}$,
  it suffices by \cref{cor:vfracglue-map} to define a pair of functions $f_\top : \ExpV{\X[\top]}{\Y[\top]}$ and $f_{\ABS} : \ExpV{\X[\ABS]}{\Y[\ABS]}$ such that
  \[\begin{tikzcd}
    {\X[\top]} & {\Y[\top]} \\
    {\X[\ABS]} & {\Y[\ABS]}
    \arrow[arrv, "{f_\top}", from=1-1, to=1-2]
    \arrow[arrv, ""{name=0, anchor=center, inner sep=0}, "\chi"', from=1-1, to=2-1]
    \arrow[arrv, ""{name=1, anchor=center, inner sep=0}, "\psi", from=1-2, to=2-2]
    \arrow[arrv, "{f_{\ABS}}"', from=2-1, to=2-2]
    \arrow["{=}"{description}, draw=none, from=1, to=0]
  \end{tikzcd}\]
  using $f_\cl = \ClosedV{f_\top}$ and $f_\op = \OpenV{f_{\ABS}}$.
\end{lemma}
The equation, rendered as a square, is a functional analogue of the relational notion of representation independence~\citep{reynolds>1983}.

\begin{example}[{\citealp[\S 2.2.1]{grodin-li-harper>2026}}]\label{ex:afat-bq}
  Consider the batched queue functional data structure, which represents a queue as a pair of lists $(l_1, l_2)$~\citep{hood-melville>1981,burton>1982,gries>1989,okasaki>1999}.
  Elements are enqueued to the ``back'' list $l_2$ and dequeued from the ``front'' list $l_1$---unless, of course, the second list is empty, in which case the dequeue operation reverses $l_2$ to replace $l_1$.
  Although the implementation type of batched queues is $\LLV$, the generic client-facing specification type describing the mathematical behavior of queues is $\LV$; the former can be converted to the latter using the function $\chi \isdef \vlam{(l_1,l_2)}{\ListVAppend{l_1}{\DEFN{reverse}~l_2}}$.
  Both perspectives, mediated by $\chi$, can be built into a single type:
  \[ \X \isdef \GlueeV{\LLV}{\LV}{\chi}. \]
  To define a function $\DEFN{enqueue} : \ExpV{\NatV}{\ExpV{\X}{\X}}$, it suffices by \cref{lem:AbstractionV-exp} to give the true code alongside an abstract mathematical model that coheres up to $\chi$.
  Defining
  \begin{align*}
    \DEFN{enqueue}_\top~n~(l_1,l_2) \isdef (l_1, \ListVCons{n}{l_2})
    &&
    \DEFN{enqueue}_{\ABS}~n~l \isdef \ListVAppend{l}{\ListVSing{n}},
  \end{align*}
  observe that the following equation holds:
  \[\begin{tikzcd}
    \LLV && \LLV \\
    \LV && \LV
    \arrow[arrv, "{\DEFN{enqueue}_\top~n}", from=1-1, to=1-3]
    \arrow[arrv, ""{name=0, anchor=center, inner sep=0}, "\chi"', from=1-1, to=2-1]
    \arrow[arrv, ""{name=1, anchor=center, inner sep=0}, "\chi", from=1-3, to=2-3]
    \arrow[arrv, "{\DEFN{enqueue}_{\ABS}~n}"', from=2-1, to=2-3]
    \arrow["{=}"{description}, draw=none, from=1, to=0]
  \end{tikzcd}\]
  These data suffice to implement $\DEFN{enqueue}$, hiding the private batched implementation under a public list-based specification.
  The $\DEFN{empty}$ and $\DEFN{dequeue}$ operations are similar.
\end{example}

By \cref{thm:vfracglue}, a programmer knows that every type contains
a concrete type $\X[\cl]$, an abstract type $\X[\op]$, and
an abstraction function $\chi_\cl : \ExpV{\X[\cl]}{\ClosedV{\X[\op]}}$.
However, the meanings of ``concrete'' and ``abstract'' depend on the semantic interpretation of the phase proposition $\AbsV$~\citep{grodin-li-harper>2026}.

\paragraph{Semantics}
\begin{samepage}
  Using a Kripke (\ie, presheaf) semantics with two worlds $\{ \ABS \vdash \top \}$, every type is interpreted as a pair of types along with a genuine function, and every function $f : \ExpV{\X}{\Y}$ is interpreted as a coherent pair of functions:
  \begin{center}
    \begin{minipage}{0.5\textwidth}
  \[ \Interp{\X} \isdef \Presheaf[\Interp[\vdash]{\X}]{\Interp[\top]{\X}}{\Interp[\ABS]{\X}} \]
    \end{minipage}%
    \begin{minipage}{0.5\textwidth}
\[\Interp{f} \isdef \begin{tikzcd}
	{\Interp[\top]{\X}} & {\Interp[\top]{\Y}} \\
	{\Interp[\ABS]{\X}} & {\Interp[\vdash]{\Y}}
	\arrow["{\Interp[\top]{f}}", from=1-1, to=1-2]
	\arrow[""{name=0, anchor=center, inner sep=0}, "{\Interp[\vdash]{\X}}"', from=1-1, to=2-1]
	\arrow[""{name=1, anchor=center, inner sep=0}, "{\Interp[\vdash]{\Y}}", from=1-2, to=2-2]
	\arrow["{\Interp[\ABS]{f}}"', from=2-1, to=2-2]
	\arrow["{=}"{description}, draw=none, from=0, to=1]
\end{tikzcd}\]
    \end{minipage}%
  \end{center}
\end{samepage}

\subsection{Potential Functions as Types: Sleator's Potential Functions, Synthetically}

In the previous section, we recalled that a function $\chi : \ExpV{\X[\top]}{\X[\ABS]}$ can be assembled into a value type.
Now, we extend this construction to the level of computation types, enabling a homomorphism $\alpha : \LolliV{\A[\top]}{\A[\ABS]}$ to be assembled into a computation type.
Viewing the cost of such homomorphisms as potential \`a la \citet{grodin-harper>2024}, this will render \emph{potential functions as types}.

With the goal of proving an analogous fracture and gluing principle for the universe $\CC$ of computation types, we must first define analogues to the abstract and concrete modalities (\cref{def:open-modality,def:closed-modality}) at the level of computation types.
First, the abstract modality adapts straightforwardly:

\begin{definition}[{\AgdaFormalized{Calf.Computation.Open}}]
  The \emph{abstract modality on computation types}
  is defined as $\OpenC{\A} \isdef \PowerC{\AbsV}{\A}$, where its unit $\etaOpC[\A] : \LolliV{\A}{\OpenC{\A}}$ is defined as $\clam{a}{\plam{\abs*}{a}}$.
  Say that a computation type $\A$ is \emph{abstract} when $\etaOpC[\A]$ is an equivalence, and write $\CC[\op]$ for the universe of abstract computation types.
\end{definition}

In \cref{def:closed-modality}, the concrete modality $\ClosedV{\X}$ on value types is defined as a pushout of $\X$ and $\AbsV$ over the product type $\ProdV{\AbsV}{\X}$.
To define the concrete modality $\ClosedC{\A}$ on computation types, we replace $\AbsV$ and $\ProdV{\AbsV}{\X}$ with the copowers $\CopowerC{\AbsV}{\UnitC}$ and $\CopowerC{\AbsV}{\A}$, respectively.

\begin{definition}[{\AgdaFormalized{Calf.Computation.Closed}}]
  The \emph{concrete modality on computation types}
  is the following pushout:
      \[\begin{tikzcd}
        {\CopowerC{\AbsV}{\A}} & {\CopowerC{\AbsV}{\UnitC}} \\
        \A & {\ClosedC{\A}}
        \arrow[mmv, "{\CopowerC{\AbsV}{\mathsf{triv}}}", from=1-1, to=1-2]
        \arrow[mmv, "{\mathsf{proj}}"', from=1-1, to=2-1]
        \arrow[mmv, "\starClC", from=1-2, to=2-2]
        \arrow[mmv, "{\etaClC[\A]}", from=2-1, to=2-2]
        \arrow["\lrcorner"{anchor=center, pos=0.125, rotate=180, color=RedDevil}, draw=none, from=2-2, to=1-1]
      \end{tikzcd}\]
  Crucially, given $\abs : \AbsV$, we have $\ClosedC{\A} = \UnitC$.
  Say that a computation type $\A$ is \emph{concrete} when $\etaClC[\A]$ is an equivalence, and write $\CC[\cl]$ for the universe of concrete computation types.
\end{definition}

Although the abstract modality $\OpenC{\A}$ is well-behaved, the concrete modality $\ClosedC{\A}$ given above is insufficient%
\footnote{This modal operator need not constitute a stable orthogonal factorization system~\citep{rijke-shulman-spitters>2020}.}
for fracture and gluing without further information about the interpretation of computation types.
Thus, it is important to reveal additional information about the structure of $\CC$.

\subsubsection{Computation Types as Cost Algebras}

In prior work on Calf, computation types are syntactically left unspecified but semantically interpreted as \emph{cost algebras}~\citep{niu-sterling-grodin-harper>2022,li-harper>2025}.
In this work, rather than leaving $\CC$ unknown, we \emph{define} $\CC$ to be the universe of cost algebras.
\begin{samepage}
\begin{definition}[{\AgdaFormalized{Calf.Computation}}]\label{def:cost-algebra}
  A \emph{cost algebra} $\A : \CC$ consists of
  \begin{enumerate}
    \item
      an underlying preordered%
      \footnote{In this setting, a preordered value type is a synthetic preorder~\citep{grodin-niu-sterling-harper>2024} that, to accommodate the higher type theory of this work, is an h-set~\citep{univalentfoundations>2013}. Both of these requirements are orthogonality conditions, rendering synthetic preorders as a reflective subuniverse~\citep{rijke-shulman-spitters>2020} of $\VV$.}
      value type $\U{\A} : \VV$ and
    \item a function $\charge[\A]{-} : \ExpV{\Cost}{\ExpV{\U{\A}}{\U{\A}}}$ such that
    \item $\charge[\A]{0}[a] = a$ and
    \item $\charge[\A]{c_1 + c_2}[a] = \charge[\A]{c_1}[\charge[\A]{c_2}[a]]$.
  \end{enumerate}
\end{definition}
\end{samepage}
In other words, a cost algebra consists of a value type along with a sensible way of incorporating cost, used to give the semantics of the cost effect as given above.%
\footnote{In follow-up work by~\citet{grodin-niu-sterling-harper>2024}, the semantics of computation types was generalized to support other effects; for simplicity, we only consider cost in this work, although the arguments naturally adapt to the setting of various other effects.}
Correspondingly, we define the value type $\LolliV{\A}{\B}$ to be the type of cost algebra homomorphisms.
\begin{samepage}
\begin{definition}[{\AgdaFormalized{Calf.Computation}}]\label{def:cost-algebra-hom}
  A \emph{cost algebra homomorphism} $f : \LolliV{\A}{\B}$ consists of
  \begin{enumerate}
    \item an underlying value-level function $\U{f} : \ExpV{\U{\A}}{\U{\B}}$ such that
    \item for all $a : \U{\A}$, it is the case that $\charge[\B]{c}[\U{f}(a)] = \U{f}(\charge[\A]{c}[a])$.
  \end{enumerate}
  For readability, we often write $f$ instead of $\U{f}$, such as $f(a)$ instead of $\U{f}(a)$.
\end{definition}
\end{samepage}
Even though such an $f : \LolliV{\A}{\B}$ is a reusable value, the notation is inspired by the idea that such a function should use its input linearly, so as not to drop or duplicate incoming effects~\citep{egger-mogelberg-simpson>2009,egger-mogelberg-simpson>2014}.
This notion of linearity is \emph{semantic}, requiring a proof that cost is preserved.

Given this revelation of computation types as cost algebras, the computation-level abstract and concrete modalities satisfy properties sufficient for proving fracture and gluing.

\begin{lemma}[{\AgdaFormalized{Calf.Computation.Closed}}]\label{lem:u-open}\label{lem:u-closed}
  For all computation types $\A$, both $\U[\OpenC{\A}] = \OpenV[\U{\A}]$ and $\U[\ClosedC{\A}] = \ClosedV[\U{\A}]$.
\end{lemma}

\begin{lemma}[{\AgdaFormalized{Calf.Computation.Closed}}]\label{lem:open-closed-lex}
  Both $\OpenC$ and $\ClosedC$ are \emph{lex}, meaning that they preserve pullbacks.
\end{lemma}

\subsubsection{Fracture and Gluing}

We now lift the fracture and gluing theorem on the universe of value types $\VV$ (\cref{thm:vfracglue}) to a fracture and gluing theorem on the universe of computation types $\CC$, which is revealed to be the universe of cost algebras.

For a concrete type $\A[\cl] : \CC[\cl]$ and an abstract type $\A[\op] : \CC[\op]$, an \emph{abstraction homomorphism} is a homomorphism $\alpha_{\cl} : \LolliV{\A[\cl]}{\ClosedC{\A[\op]}}$.
Beyond the abstraction capabilities of an abstraction function $\chi_\cl : \ExpV{\X[\cl]}{\ClosedV{\X[\op]}}$, an abstraction homomorphism may perform the cost effect, which serves the role of potential.
In this way, the following theorem explains why every computation type $\A : \CC$ contains not only an abstraction function but also a potential function.

\begin{theorem}[Fracture and Gluing of $\CC$, {\AgdaFormalized{Calf.Computation.Glue}}]\label{thm:cfracglue}
  Every computation type $\A : \CC$ contains exactly the data of a concrete type $\A[\cl]$, an abstract type $\A[\op]$, and an abstraction homomorphism between them:
  \[
    \CC = \DSumV{\A[\cl] : \CC[\cl]}{\DSumV{\A[\op] : \CC[\op]}{\LolliV{\A[\cl]}{\ClosedC{\A[\op]}}}}.
  \]
\end{theorem}
\begin{proof}[Proof Sketch]
  Fracture a type $\A$ into $(\ClosedC{\A}, \OpenC{\A}, \ClosedC{\etaOpC})$, and glue $(\A[\cl], \A[\op], \alpha_\cl)$ into the pullback
  \[ \GlueC{\A[\cl]}{\A[\op]}{\alpha_\cl} \isdef \ProdC{\A[\cl]}{\A[\op]}[\ClosedC{\A[\op]}], \]
  following \citet{rijke-shulman-spitters>2020}.
  To show that $\A = \GlueC{\ClosedC{\A}}{\OpenC{\A}}{\ClosedC{\etaOpC}}$, using univalence, it suffices to define a map $f : \LolliV{\A}{\GlueC{\ClosedC{\A}}{\OpenC{\A}}{\ClosedC{\etaOpC}}}$ and prove that it is an equivalence.
  Let
  \begin{align*}
    &f : \LolliV{\A}{\GlueC{\ClosedC{\A}}{\OpenC{\A}}{\ClosedC{\etaOpC}}} \\
    &f~a \isdef (\etaClC a, \etaOpC a)
  \end{align*}
  using the modal units $\etaClC : \LolliV{\A}{\ClosedC{\A}}$ and $\etaOpC : \LolliV{\A}{\OpenC{\A}}$.
  To show $f$ is an equivalence, it suffices to show $\U{f} : \ExpV{\U{\A}}{\U[\GlueC{\ClosedC{\A}}{\OpenC{\A}}{\ClosedC{\etaOpC}}]}$ is an equivalence because $\U : \ExpV{\CC}{\VV}$ is \emph{conservative}.
  As
  \begin{align*}
    \U[\GlueC{\ClosedC{\A}}{\OpenC{\A}}{\ClosedC{\etaOpC}}]
      &= \U[\ProdC{\ClosedC{\A}}{\OpenC{\A}}[\ClosedC{\OpenC{\A}}]] \\
      &= \ProdV{\U[\ClosedC{\A}]}{\U[\OpenC{\A}]}[\U[\ClosedC{\OpenC{\A}}]]  \tag*{(right adjoints preserve limits)} \\
      &= \ProdV{\ClosedV[\U{\A}]}{\OpenV[\U{\A}]}[\ClosedV{\OpenV[\U{\A}]}]  \tag*{(\cref{lem:u-closed,lem:u-open})}
  \end{align*}
  renders $\U{f}$ as the fracture-and-gluing of the value type $\U{\A}$, this map is an equivalence by the fracture and gluing theorem for $\VV$ (\cref{thm:vfracglue}).
  In the other direction, we have that
  \begin{align*}
    &\ClosedC[\GlueC{\A[\cl]}{\A[\op]}{\ClosedC{\etaOpC}}]  &&\OpenC[\GlueC{\A[\cl]}{\A[\op]}{\ClosedC{\etaOpC}}] \\
    &= \ClosedC[\ProdC{\A[\cl]}{\A[\op]}[\ClosedC{\A[\op]}]] &&= \OpenC[\ProdC{\A[\cl]}{\A[\op]}[\ClosedC{\A[\op]}]]  \\
    &= \ProdC{\ClosedC{\A[\cl]}}{\ClosedC{\A[\op]}}[\ClosedC{\ClosedC{\A[\op]}}] &&= \ProdC{\OpenC{\A[\cl]}}{\OpenC{\A[\op]}}[\OpenC{\ClosedC{\A[\op]}}]  \tag*{(\cref{lem:open-closed-lex})} \\
    &= \ProdC{\ClosedC{\A[\cl]}}{\ClosedC{\A[\op]}}[\ClosedC{\A[\op]}] &&= \ProdC{\UnitC}{\OpenC{\A[\op]}}[\UnitC] \\
    &= \ClosedC{\A[\cl]} &&= \OpenC{\A[\op]} \\
    &= \A[\cl] &&= \A[\op]
  \end{align*}
  using the fact that $\OpenC{\ClosedC{\A}} = \UnitC$.
\end{proof}
\begin{corollary}[\AgdaFormalized{Calf.Computation.Abstraction}]\label{cor:cfracglue-map}
  Every $f : \LolliV{\A}{\B}$ consists of a concrete homomorphism $f_\cl : \LolliV{\ClosedC{\A}}{\ClosedC{\B}}$ and an abstract homomorphism $f_\op : \LolliV{\OpenC{\A}}{\OpenC{\B}}$ that cohere (formally, $\ClosedC{\etaOpC[\B]} \circ f_\cl = \ClosedC[f_\op \circ \etaOpC[\A]]$).
\end{corollary}

While the fracture and gluing result above depends centrally on modalities, much of the remainder of this work can be achieved without the direct use of modalities.
Analogous to \cref{def:AbstractionV}, it is possible to treat any homomorphism as an abstraction homomorphism.
\begin{definition}[\AgdaFormalized{Calf.Computation.Abstraction}]\label{def:AbstractionC}
  Define $\GlueeC{\A[\top]}{\A[\ABS]}{\alpha} \isdef \GlueC{\ClosedC{\A[\top]}}{\OpenC{\A[\ABS]}}{\ClosedC{(\etaOpC \circ \alpha)}}$.
\end{definition}
\begin{corollary}[\AgdaFormalized{Calf.Computation.Abstraction}]\label{cor:AbstractionC-id}
  By \cref{thm:cfracglue}, it is the case that $\A = \GlueeC{\A}{\A}{\idC[\A]}$ for all $\A$.
\end{corollary}
\begin{lemma}[{\AgdaFormalized{Calf.Computation.Abstraction}}]\label{lem:AbstractionC-lolli}
  To define a map
  $f : \LolliV{\GlueeC{\A[\top]}{\A[\ABS]}{\alpha}}{\GlueeC{\B[\top]}{\B[\ABS]}{\beta}}$,
  it suffices by \cref{cor:cfracglue-map} to define a pair of maps $f_\top : \LolliV{\A[\top]}{\B[\top]}$ and $f_{\ABS} : \LolliV{\A[\ABS]}{\B[\ABS]}$ such that
  \[\begin{tikzcd}
    {\A[\top]} & {\B[\top]} \\
    {\A[\ABS]} & {\B[\ABS]}
    \arrow[mmv, "{f_\top}", from=1-1, to=1-2]
    \arrow[mmv, ""{name=0, anchor=center, inner sep=0}, "\alpha"', from=1-1, to=2-1]
    \arrow[mmv, ""{name=1, anchor=center, inner sep=0}, "\beta", from=1-2, to=2-2]
    \arrow[mmv, "{f_{\ABS}}"', from=2-1, to=2-2]
    \arrow["{=}"{description}, draw=none, from=1, to=0]
  \end{tikzcd}\]
  using $f_\cl = \ClosedC{f_\top}$ and $f_\op = \OpenC{f_{\ABS}}$.
\end{lemma}

The payoff of this construction, building on the insight of \citet{grodin-harper>2024}, is that every type $\A$ contains an abstraction homomorphism indicating how much potential is stored within.
Moreover, every homomorphism $\LolliV{\A}{\B}$ contains proofs of both abstraction and the conservation of potential, \emph{synthetically} reconstructing the physicist's view of amortized analysis.

\begin{example}[{\AgdaFormalized{Examples.Queue}}]\label{ex:pfat-bq}
  Adapting \cref{ex:afat-bq}, we now build a cost-aware, ephemeral variant of the batched queue data structure that includes both abstraction and potential.
  The homomorphism that abstracts a batched queue $(l_1, l_2)$ as a single list $\chi~(l_1, l_2) \isdef \ListVAppend{l_1}{\DEFN{reverse}~l_2}$ and emits its potential $\Phi~(l_1, l_2) \isdef \ListVLength{l_2}$ can be constructed as follows:
  \begin{align*}
    &\alpha : \LolliV{\F[\LLV]}{\F[\LV]} \\
    &\alpha~(\ret[l_1, l_2]) \isdef \charge{\Phi~(l_1, l_2)}[\ret[\chi~(l_1, l_2)]]
  \end{align*}
  From this homomorphism $\alpha$, we may build a type that contains all of this data, including the pair-of-lists implementation type, the single-list specification type, the value-level function $\chi$, and the potential function $\Phi$:
  \[ \A \isdef \GlueeC{\F[\LLV]}{\F[\LV]}{\alpha}. \]
  Using this type $\A$, we may implement queue operations; each consisting of a concrete aspect (on the pair-of-lists type, annotated with realistic costs) and an abstract aspect (on the single-list type, annotated with amortized costs), linked by $\alpha$ in both behavior (up to $\chi$) and cost (up to $\Phi$).

  To define $\DEFN{enqueue} : \ExpV{\NatV}{\LolliV{\A}{\A}}$, it suffices by \cref{lem:AbstractionC-lolli} to give the true code $\DEFN{enqueue}_\top$ and an abstract mathematical model $\DEFN{enqueue}_{\ABS}$ that cohere up to $\alpha$, the proof of which includes both abstraction and amortization.
  Assuming a cost model that counts recursive calls, the true enqueue algorithm is defined as
  \begin{align*}
    &\DEFN{enqueue}_\top : \ExpV{\NatV}{\LolliV{\F[\LLV]}{\F[\LLV]}}
    \\
    &\DEFN{enqueue}_\top~n~(\ret[l_1, l_2]) \isdef \ret[l_1, \ListVCons{n}{l_2}],
  \end{align*}
  without any cost annotations.
  On the other hand, the abstract specification is defined as
  \begin{align*}
    &\DEFN{enqueue}_{\ABS} : \ExpV{\NatV}{\LolliV{\F[\LV]}{\F[\LV]}}
    \\
    &\DEFN{enqueue}_{\ABS}~n~(\ret{l}) \isdef \charge{1}[\ret[\ListVAppend{l}{\ListVSing{n}}]],
  \end{align*}
  describing a client-facing \emph{amortized cost interface} with one unit of cost to cohere with the potential function (and in anticipation of an impending linear-cost dequeue).
  The remainder of the enqueue implementation is the proof that the true code coheres with the abstract amortized specification:
\[\begin{tikzcd}
	{\F[\LLV]} && {\F[\LLV]} \\
	{\F[\LV]} && {\F[\LV]}
	\arrow[mm, "{\DEFN{enqueue}_\top~n}", from=1-1, to=1-3]
	\arrow[mm, ""{name=0, anchor=center, inner sep=0}, "\alpha"', from=1-1, to=2-1]
	\arrow[mm, ""{name=1, anchor=center, inner sep=0}, "\alpha", from=1-3, to=2-3]
	\arrow[mm, "{\DEFN{enqueue}_{\ABS}~n}"', from=2-1, to=2-3]
	\arrow["{=}"{description}, draw=none, from=0, to=1]
\end{tikzcd}\]
  Equationally, the proof obligation expressed by this square holds by the following reasoning:
  \begin{align*}
    \alpha~(\DEFN{enqueue}_\top~n~(\ret[l_1,l_2]))
      &= \alpha~(\ret[l_1, \ListVCons{n}{l_2}]) \\
      &= \charge{\ListVLength{\ListVCons{n}{l_2}}}[\ret[\ListVAppend{l_1}{\DEFN{reverse}~(\ListVCons{n}{l_2})}]] \\
      &= \charge{\ListVLength{l_2} + 1}[\ret[\ListVAppend{l_1}{\ListVAppend{\DEFN{reverse}~l_2}{\ListVSing{n}}}]] \\
      &= \charge{\ListVLength{l_2}}[\charge{1}[\ret[\ListVAppend{l_1}{\ListVAppend{\DEFN{reverse}~l_2}{\ListVSing{n}}}]]] \\
      &= \charge{\ListVLength{l_2}}[\DEFN{enqueue}_{\ABS}~n~(\ret[\ListVAppend{l_1}{\DEFN{reverse}~l_2}])] \\
      &= \DEFN{enqueue}_{\ABS}~n~(\charge{\ListVLength{l_2}}[\ret[\ListVAppend{l_1}{\DEFN{reverse}~l_2}]]) \tag*{$(\star)$} \\
      &= \DEFN{enqueue}_{\ABS}~n~(\alpha~(\ret[l_1,l_2]))
  \end{align*}
  The indicated step $(\star)$ holds because $\DEFN{enqueue}_{\ABS}$ is a homomorphism of cost algebras.
  Notice that in addition to showing that $\chi$ is preserved for abstraction \citep{grodin-li-harper>2026}, this derivation verifies that potential is conserved \citep{grodin-harper>2024}: specifically, it includes the fact that
  \[ 0 + \Phi(l_1, \ListVCons{n}{l_2}) = \Phi(l_1, l_2) + 1, \]
  where $0$ is the true cost and $1$ is the amortized cost.
  The $\DEFN{empty} : \U{\A}$ and $\DEFN{dequeue} : \LolliV{\A}{\CopowerC{\NatV}{\A}}$ operations can be similarly defined, following \citet{grodin-li-harper>2026}.
\end{example}

\subsection{Potential and Abstraction, Independently}\label{sec:physicist:PotentialC}

Although the abstraction homomorphisms built into types generally include both abstraction and potential, these concerns need not be considered simultaneously as above.
As an additional convenience, it is possible to build a potential function $\Phi : \ExpV{\X}{\Cost}$ into the type $\F{\X}$.
\begin{definition}[{\AgdaFormalized{Calf.Computation.Potential}}]\label{def:PotentialC}
  To render a potential function $\Phi : \ExpV{\X}{\Cost}$ as a type, define
  \[ \PotentialC{\X}{\Phi} \isdef \GlueeC{\F{\X}}{\F{\X}}{\varphi(\Phi)} \]
  where $\varphi(\Phi) \isdef \clam{(\ret{x})}{\charge{\Phi(x)}[\ret{x}]}$.
\end{definition}
\begin{lemma}[{\AgdaFormalized{Calf.Computation.Potential}}]\label{lem:hom-pot}
  Let $\A \isdef \PotentialC{\X}{\Phi_{\X}}$ and $\B \isdef \PotentialC{\Y}{\Phi_{\Y}}$.
  To define a homomorphism of type $\LolliV{\A}{\B}$, it suffices to provide a function $f : \ExpV{\X}{\Y}$, a true cost function $c_\top : \ExpV{\X}{\Cost}$, and an amortized cost function $c_{\ABS} : \ExpV{\X}{\Cost}$ such that potential is conserved:
  \[ c_{\top}(x) + \Phi_{\Y}(f(x)) = \Phi_{\X}(x) + c_{\ABS}(x). \]
\end{lemma}
\begin{proof}
  By \cref{lem:AbstractionC-lolli}, with $f_i~(\ret{x}) \isdef \charge{c_i(x)}[\ret[f(x)]]$.
\end{proof}

\begin{corollary}[{\AgdaFormalized{Calf.Computation.Abstraction}}]\label{cor:lolli-gluee}
  Let $\B \isdef \GlueeC{\B[\top]}{\B[\ABS]}{\beta}$ and let $\A : \CC$ be arbitrary.
  \begin{enumerate}
    \item To define a homomorphism of type $f : \LolliV{\A}{\B}$, it suffices to provide $f_\top : \LolliV{\A}{\B[\top]}$.
    \item To define a homomorphism of type $g : \LolliV{\B}{\A}$, it suffices to provide $g_{\ABS} : \LolliV{\B[\ABS]}{\A}$.
  \end{enumerate}
\end{corollary}
\begin{proof}
  By \cref{cor:AbstractionC-id}, it is the case that $\A = \GlueeC{\A}{\A}{\idC[\A]}$.
  Then, the results follow by \cref{lem:AbstractionC-lolli}, with $f_{\ABS} \isdef \beta \circ f_\top$ and $g_\top \isdef g_{\ABS} \circ \beta$.
\end{proof}

Using this construction, the amortized analysis and abstraction of \cref{ex:pfat-bq} may be achieved sequentially rather than simultaneously.

\begin{example}\label{ex:sequential-bq}
  First, define a type to incorporate only the potential function for batched queues:
  \[ \B \isdef \PotentialC{\LLV}{\vlam{(l_1, l_2)}{\ListVLength{l_2}}}. \]
  To define $\DEFN{enqueue}_\top : \ExpV{\NatV}{\LolliV{\B}{\B}}$, it suffices by \cref{lem:hom-pot} to define
  \begin{align*}
    f~n~(l_1,l_2) \isdef (l_1, \ListVCons{n}{l_2})
    &&
    c_\top~n~(l_1,l_2) \isdef 0
    &&
    c_{\ABS}~n~(l_1,l_2) \isdef 1
  \end{align*}
  and prove that the conservation of potential equation holds.
  This definition exports an amortized cost specification saying that
  the enqueue operation takes $1$ amortized cost; however, the behavioral aspect of the specification still reveals that the data structure is implemented as a pair of lists.
  To remedy this, we define a type whose concrete part is inherited from $\B$ itself, but whose abstract part exports the list-based representation (and maintains the amortized cost specification).
  Let
  \[ \A \isdef \GlueeC{\B}{\F[\LV]}{\alpha}, \]
  where $\alpha$ can be defined using \cref{cor:lolli-gluee} given only the lifting of the abstraction function
  \[ \vlam{(l_1,l_2)}{\ListVAppend{l_1}{\DEFN{reverse}~l_2}} : \ExpV{\LLV}{\LV}. \]
  To define $\DEFN{enqueue} : \ExpV{\NatV}{\LolliV{\A}{\A}}$, it suffices to combine $\DEFN{enqueue}_\top$ from above with $\DEFN{enqueue}_{\ABS}$ from \cref{ex:pfat-bq}.
  Thus, we were able to first perform the amortized analysis of batched queues and subsequently overlay a coherent abstract mathematical data specification.
\end{example}

\section{Loss of Energy Due to Abstraction}\label{sec:modularity}

As developed by \citet[\S 3]{grodin-li-harper>2026}, the abstract phase facilitates modularity, guaranteeing that client verifications do not depend on library implementation details.
In this section, we apply such reasoning to the present setting to facilitate modular verification of amortized cost.
Then, building on the approach of \citet[\S 4.5]{grodin-li-harper>2026}, we show how to accommodate inequality in the conservation of potential condition, which improves modularity in the cost-aware setting.

\subsection{Modularity and Amortized Cost Interfaces}

The abstract phase makes it possible to uniformly isolate the public-facing specification associated with any type or program, allowing correctness theorems to be stated under the explicit assumption that concrete implementation details have been ignored.
This facilitates modular verification: if clients only prove theorems under the assumption of $\AbsV$, there is no trouble modifying a library so long as it continues to maintain the same public-facing specification.
Phase-sensitive restrictions can be made using the notion of a \emph{specification type}.

\NewDocumentCommand{\PV}{}{\VAL*{\P}}
\begin{definition}[Specification Type, \citealp{grodin-li-harper>2026}]\label{def:spec-type}
  Let $\PV : \VV$ be a proposition.
  The \emph{$\PV$-phase specification type} for a specification $x_\circ : \ExpV{\PV}{\X}$ is the type
  \[ \SPEC{\PV}{\X}{x_\circ} \isdef \DSumV{x : \X}{(\DProdV{\abs : \PV}{x = x_\circ(\abs)})} \]
  describing all the inhabitants $x : \X$ that cohere with $x_\circ$ in the phase $\PV$.%
  \footnote{Per \citet{grodin-li-harper>2026}, the notation is inspired by extension types~\citep{riehl-shulman>2017}, but the equality is not definitional.}
\end{definition}

Using a specification type, we will define an interface of ephemeral queues.
Although implementations may perform the cost effect, the interface for ephemeral queues should only restrict abstract \emph{behavior}.
This can be achieved using the \emph{behavioral phase} proposition $\BehV$ of Calf~\citep{niu-sterling-grodin-harper>2022}, which isolates the behavior of a program from its cost by erasing the cost effect:
\[ \ExpV{\BehV}{\charge{c}[e] = e}. \]
When combined with abstraction, $\BehV$ is axiomatized to imply $\AbsV$~\citep{grodin-li-harper>2026}.

\begin{example}[Ephemeral Queue Interface]\label{ex:ephemeral-queue}
  Say an ephemeral pre-queue is a computation type $\A : \CC$ equipped with standard queue operations:
  \[ \VTY{PreQueue}_{\CC} \isdef \DSumV{\A : \CC}{\ProdV{(\LABEL{empty} : \U{\A})}{\ProdV{(\LABEL{enqueue} : \ExpV{\NatV}{\LolliV{\A}{\A}})}{(\LABEL{dequeue} : \LolliV{\A}{\CopowerC{\NatV}{\A}})}}}. \]
  The interface of ephemeral queues, $\VTY{Queue}_{\CC}$, refines $\VTY{PreQueue}_{\CC}$ with a \emph{behavioral} specification:
  \[ \SPEC{\BehV}{\VTY{PreQueue}_{\CC}}{(\F[\LV], \ret{\ListVNil}, (\vlam{n}{\clam{(\ret{l})}{\ret[\ListVAppend{l}{\ListVSing{n}}]}}), \DEFN{uncons})}. \]
  This restriction completely determines the abstract aspect of an implementation, aside from cost.
  The batched queues of \cref{ex:pfat-bq} are a valid implementation of this interface, by construction: erasing the costs of the interface-level components recovers precisely the given specification.
\end{example}

Although this behavioral restriction entirely pins down the behavior of an implementation, it makes no claims about the cost of an implementation.
In order to reveal information about the cost, it suffices to \emph{refine} the interface $\VTY{Queue}_{\CC}$ with stronger guarantees, such as an $\AbsV$-phase restriction revealing an \emph{amortized cost specification}.

\begin{example}[Ephemeral Queue Cost Interface \AgdaFormalized{Examples.Queue}]\label{ex:ephemeral-queue-cost}
  Building on the ephemeral queue interface of \cref{ex:ephemeral-queue}, consider the following refinement, noting the use of $\AbsV$ instead of $\BehV$:
  \[ \SPEC{\AbsV}{\VTY{PreQueue}}{(\F[\LV], \ret{\ListVNil}, (\vlam{n}{\clam{(\ret{l})}{\charge{1}[\ret[\ListVAppend{l}{\ListVSing{n}}]]}}), \DEFN{uncons})}. \]
  Beyond the behavioral guarantees consistent with $\VTY{Queue}_{\CC}$, this interface exports amortized costs of the operations, classifying queues whose enqueue operation has amortized cost of $1$ and whose empty and dequeue operations (implicitly, by lack of cost annotation) have amortized cost of $0$.
  By construction, the batched queue implementation of \cref{ex:pfat-bq} inhabits this refinement.
\end{example}

\subsection{Amortized Upper Bounds as Lax Commutative Squares}

The batched queue data structure has the property that its amortized cost model exactly matches the implementation.
In particular, its analysis satisfies a strict conservation principle:
\[ c_\top(x) + \Phi(f(x)) = \Phi(x) + c_{\ABS}(x). \]
However, such exactness is rare: the true cost of an algorithm often depends on private implementation details, and the amortized cost is thus merely an upper bound of the true cost.
To this end, it is necessary to relax the strict equality-based commutativity of the squares built into homomorphisms to merely a lax inequality.
This representing the idea of energy not being perfectly conserved:
\[ c_\top(x) + \Phi(f(x)) \le \Phi(x) + c_{\ABS}(x). \]
In physics, energy can be lost due to sources such as friction; in this setting, amortized cost bounds can be weakened due to the demands of \emph{abstraction}.

In order to support weakening of cost bounds via lax commutative squares, we follow \citet{grodin-li-harper>2026} and use the \emph{sealing monad} $\SealC$, defined here on computation types as a comma object:
\begin{center}
  \begin{minipage}{0.5\textwidth}
\begin{align*}
  &\SealC : \CC \to \CC \\
  &\SealC{\A} = \CommaC{\ClosedC{\A}}{\OpenC{\A}}{\ClosedC{\OpenC{\A}}}
\end{align*}
  \end{minipage}%
  \begin{minipage}{0.5\textwidth}
\[\begin{tikzcd}
	{\SealC{\A}} & {\ClosedC{\A}} \\
	{\OpenC{\A}} & {\ClosedC{\OpenC{\A}}}
	\arrow[mmv, from=1-1, to=1-2]
	\arrow[mmv, ""{name=0, anchor=center, inner sep=0}, from=1-1, to=2-1]
	\arrow[mmv, ""{name=1, anchor=center, inner sep=0}, "{\ClosedC{\etaOpC}}", from=1-2, to=2-2]
	\arrow[mmv, "\etaClC"', from=2-1, to=2-2]
	\arrow["\lrcorner"{anchor=center, pos=0.125}, draw=none, from=1-1, to=2-2]
	\arrow["\ge"{description}, draw=none, from=1, to=0]
\end{tikzcd}\]
  \end{minipage}%
\end{center}
Semantically%
\footnote{In the presheaf semantics, the synthetic sealing monad is interpreted as the free opfibration 2-monad~\citep{street>1974}.}%
, a Kleisli map $f : \LolliV{\A}{\SealC{\B}}$ is exactly a lax commutative square:
\[\begin{tikzcd}
	{\Interp[\top]{\A}} & {\Interp[\top]{\B}} \\
	{\Interp[\ABS]{\A}} & {\Interp[\ABS]{\B}}
	\arrow[mm, "{\Interp[\top]{f}}", from=1-1, to=1-2]
	\arrow[mm, ""{name=0, anchor=center, inner sep=0}, "{\Interp[\vdash]{\A}}"', from=1-1, to=2-1]
	\arrow[mm, ""{name=1, anchor=center, inner sep=0}, "{\Interp[\vdash]{\B}}", from=1-2, to=2-2]
	\arrow[mm, "{\Interp[\ABS]{f}}"', from=2-1, to=2-2]
	\arrow[mm, "{\ge}"{description}, draw=none, from=0, to=1]
\end{tikzcd}\]
Thus, we define a type of lax homomorphisms%
\footnote{Note that the laxity here is of the square induced by the abstract phase, \emph{not} of the cost algebra homomorphism.}
that allows weakening of cost and potential:
\[ \LolliV*{\A}{\B} \isdef \LolliV{\A}{\SealC{\B}}. \]

\begin{remark}
  \citet{grodin-li-harper>2026} added a sealing effect to computation types (at the same level as the cost effect) and used the sealing monad as a semantics.
  However, this interferes with our proof of fracture and gluing (\cref{thm:cfracglue}); hence, we make the sealing monad user-specified.
\end{remark}

\begin{lemma}\label{lem:seal}
  In the abstract phase (\ie, assuming $\AbsV$), it is the case that $\SealC{\A} = \A$.
\end{lemma}

\begin{example}
  The splay tree data structure~\citep{sleator-tarjan>1985} is a purely functional implementation of arrays with an amortized logarithmic-time lookup operation.
  The precise cost depends on internal implementation details and is not exactly representable in the abstract phase; to accommodate this laxity, one may use a lax homomorphism, as verified analytically in Calf by \citet{kebuladze>2025}.
  In the present synthetic setting, the implementation type of splay trees can be
  \[ \A \isdef \SealC[\GlueeC{\F[\VTY{Tree}~\NatV]}{\F[\LV]}{\alpha}], \]
  where $\alpha$ is the abstraction homomorphism described by \citeauthor{kebuladze>2025}.
  By \cref{lem:seal}, this implementation type is revealed as $\F[\LV]$ in the abstract phase, facilitating modularity.
\end{example}

\section{The Banker's View}\label{sec:banker}

In contrast to the physicist's potential, the banker imagines \emph{credits} and \emph{debits} being stored within data structures.
In the original work on amortized analysis, \citet{tarjan>1985} notes that credits and debits will not appear in programs themselves.
However, in the present setting of cost verification, credits and debits will indeed appear both in types and terms.
To render the banker's view within type theory, we will define notions of credit and debit internally in terms of abstraction.

\subsection{Credits and Debits}

We now implement credit and debit within the type theory using abstraction homomorphisms.

\subsubsection{The Credit Operator}

Because potential is analogous to credit, we may store $c$ credits on a data structure by emitting $c$ additional units of cost in the included abstraction homomorphism.

\begin{definition}[{\AgdaFormalized{Calf.Computation.Credit}}]
  The \emph{credit operator} annotates a type $\A$ with $c$ additional credits by constructing a variant of $\A$ with an abstraction homomorphism that emits $c$ additional units of cost:
  \begin{align*}
    &\CMP*{\vartriangleright} : \ExpV{\Cost}{\ExpV{\CC}{\CC}} \\
    &\Credit{c}{\A} \isdef \GlueeC{\A}{\A}{\charge[\A]{c}}
  \end{align*}
\end{definition}

\begin{remark}
  By \cref{def:cost-algebra}, we have $\charge[\A]{c} : \ExpV{\U{\A}}{\U{\A}}$.
  The ability to lift the cost effect to a homomorphism $\LolliV{\A}{\A}$ is precisely commutativity of the cost monoid $(\Cost, 0, +)$.
\end{remark}

\begin{remark}
  The credit operator has the effect not of replacing the existing abstraction homomorphism contained within $\A$, but rather increasing its cost by $c$.
  In the Kripke semantics:
  \begin{align*}
    \Credit{c}{\PresheafC[\alpha]{\Interp[\top]{\A}}{\Interp[\ABS]{\A}}} \isdef \PresheafC[\charge[\A[\top]]{c} \;\circ\; \alpha]{\Interp[\top]{\A}}{\Interp[\ABS]{\A}}
  \end{align*}
\end{remark}

\begin{lemma}[{\AgdaFormalized{Calf.Computation.Credit}}]\label{lem:credit-modalities}
  Credits are \emph{ghost data}, invisible in both the concrete program and the abstract specification, only appearing to mediate between the two:
  $\ClosedC[\Credit{c}{\A}] = \ClosedC{\A}$ and $\OpenC[\Credit{c}{\A}] = \OpenC{\A}$.
\end{lemma}

\begin{lemma}\label{lem:credit-functorial}
  For $c' \ge c$, the credit operator admits a weakening map $\LolliV*{\Credit{c'}{\A}}{\Credit{c}{\A}}$ that wastes credits.
\end{lemma}

\begin{lemma}[{\AgdaFormalized{Calf.Computation.Credit}}]\label{lem:credit-monoidal}
  The credit operator
  satisfies
  $\Credit{0}{\A} = \A$ and $\Credit{c_1 + c_2}{\A} = \Credit{c_1}{\Credit{c_2}{\A}}$.
\end{lemma}

\NewDocumentCommand{\CreditSpend}{m o}{\LABEL{spend}\langle #1 \rangle\IfValueT{#2}{({#2})}}
\NewDocumentCommand{\CreditSave}{m o}{\LABEL{save}\langle #1 \rangle\IfValueT{#2}{({#2})}}

\begin{definition}[{\AgdaFormalized{Calf.Computation.Credit}}]\label{def:spend}
  Using \cref{cor:lolli-gluee} with $f_\top, g_{\ABS} : \LolliV{\A}{\A}$ being the identity function, define
  \begin{align*}
    \CreditSave{c} : \LolliV{\A}{\Credit{c}{\A}}  && \CreditSpend{c} : \LolliV{\Credit{c}{\A}}{\A}.
  \end{align*}
\end{definition}

\begin{samepage}
The $\CreditSave{c}$ operation incurs no true cost, but it instead reports $c$ units of abstract cost to clients in the abstract phase.
Dually, the $\CreditSpend{c}$ operation truly incurs $c$ cost, but it reports no cost to clients in the abstract phase.
This can be visualized in the Kripke semantics, where the $\CreditSave{c}$ and $\CreditSpend{c}$ operations are interpreted as the following vertical composites of squares, respectively:
\begin{center}
  \begin{minipage}{0.5\textwidth}
\[\begin{tikzcd}
	{\Interp[\top]{\A}} & {\Interp[\top]{\A}} \\
	{\Interp[\ABS]{\A}} & {\Interp[\ABS]{\A}} \\
	{\Interp[\ABS]{\A}} & {\Interp[\ABS]{\A}}
	\arrow[mm, r,-,double equal sign distance,double, from=1-1, to=1-2]
	\arrow[mm, ""{name=0, anchor=center, inner sep=0}, "\alpha"', from=1-1, to=2-1]
	\arrow[mm, ""{name=1, anchor=center, inner sep=0}, "\alpha", from=1-2, to=2-2]
	\arrow[mm, r,-,double equal sign distance,double,, from=2-1, to=2-2]
	\arrow[mm, r,-,double equal sign distance,double,, ""{name=2, anchor=center, inner sep=0}, from=2-1, to=3-1]
	\arrow[mm, ""{name=3, anchor=center, inner sep=0}, "{\charge{c}}", from=2-2, to=3-2]
	\arrow[mm, "{\charge{c}}"', from=3-1, to=3-2]
	\arrow["{=}"{description}, draw=none, from=1, to=0]
	\arrow["{=}"{description}, draw=none, from=2, to=3]
\end{tikzcd}\]
  \end{minipage}%
  \begin{minipage}{0.5\textwidth}
\[\begin{tikzcd}
	{\Interp[\top]{\A}} & {\Interp[\top]{\A}} \\
	{\Interp[\ABS]{\A}} & {\Interp[\ABS]{\A}} \\
	{\Interp[\ABS]{\A}} & {\Interp[\ABS]{\A}}
	\arrow[mm, "{\charge{c}}", from=1-1, to=1-2]
	\arrow[mm, ""{name=0, anchor=center, inner sep=0}, "\alpha"', from=1-1, to=2-1]
	\arrow[mm, ""{name=1, anchor=center, inner sep=0}, "\alpha", from=1-2, to=2-2]
	\arrow[mm, "{\charge{c}}", from=2-1, to=2-2]
	\arrow[mm, ""{name=2, anchor=center, inner sep=0}, "{\charge{c}}"', from=2-1, to=3-1]
	\arrow[mm, r,-,double equal sign distance,double, ""{name=3, anchor=center, inner sep=0}, from=2-2, to=3-2]
	\arrow[mm, r,-,double equal sign distance,double, from=3-1, to=3-2]
	\arrow["{=}"{description}, draw=none, from=1, to=0]
	\arrow["{=}"{description}, draw=none, from=2, to=3]
\end{tikzcd}\]
  \end{minipage}%
\end{center}
Reading the diagrams from left to the right, saving credits moves cost from the abstract specification into the output type, and spending credits moves cost from the input type to the true code.
\end{samepage}

\begin{lemma}[{\AgdaFormalized{Calf.Computation.Credit}}]\label{lem:save-spend}
  For all $c$, $\CreditSpend{c} \circ \CreditSave{c} = \charge[\A]{c}$.  %
\end{lemma}
\begin{remark}
  In the language of double category theory, the bottom halves of the above composites with \cref{lem:save-spend} render cost and potential/credits as \emph{companions}~\citep{grandis-pare>2004}, representing the idea that cost and potential/credits are the same idea but in different ``dimensions''.
\end{remark}

\begin{lemma}\label{lem:PotentialC}
  Recalling \cref{def:PotentialC}, the following types are equivalent:
  \[ \PotentialC{\X}{\Phi} = \CopowerC{(x : \X)}{\Credit{\Phi(x)}{\TopC}}. \]
  This formally connects the physicist's view and the banker's view: the type containing the potential function $\Phi$ is equivalent to the type representing a value $x : \X$ stored alongside $\Phi(x) : \Cost$ credits.
\end{lemma}

\subsubsection{The Debit Operator}

Using the credit operator, we may define its dual, the debit operator, which as an input indicates that some credits are owed.

\begin{definition}[{\AgdaFormalized{Calf.Computation.Debit}}]
  The \emph{debit operator} annotates a type $\A$ with an opportunity to make use of $c$ credits, achieved via a lax homomorphism assuming credits:
  \begin{align*}
    &\CMP*{\vartriangleleft} : \ExpV{\Cost}{\ExpV{\CC}{\CC}} \\
    &\Debit{c}{\A} \isdef \LolliC*{\Credit{c}{\TopC}}{\A}
  \end{align*}
  We choose to use a lax homomorphism here in order to make debits provide an opportunity rather than a burden; with a strict homomorphism $\LolliC{\Credit{c}{\TopC}}{\A}$, all assumed credits must be spent.
\end{definition}

\begin{lemma}[{\AgdaFormalized{Calf.Computation.Debit}}]\label{lem:credit-debit-adjoint}
  For all $c : \Cost$,
  the debit and credit operators are adjoint: $\Credit{c} \dashv \Debit{c}$.
  The unit $\LolliV{\A}{\Debit{c}{\Credit{c}{\A}}}$ takes out a loan,
  and the counit $\LolliV*{\Credit{c}{\Debit{c}{\A}}}{\A}$ pays off a loan.
\end{lemma}

  Using debits, it is possible to implement advanced amortized data structures, such as the implicit queues of \citet[\S11]{okasaki>1999} as shown by \citet{danielsson>2008} and \citet{rajani>thesis}.

\subsection{Credit-Carrying Lists}\label{sec:banker:plisti}

Using the credit operator, we may build inductive data structures that contain credits, as is standard in the banker's method.
First, we develop two variants of lists: one that stores credits linear in the length of the list, and one that stores credits quadratic in the length of the list.

\begin{samepage}
\begin{definition}[{\AgdaFormalized{Calf.Computation.PList1}}]\label{lem:plisti}
  Let $\PListI{c}{\A} \isdef \ListC[\Credit{c}{\A}]$ classify lists of elements of type $\A$ in which each element is accompanied by $c$ credits.
\end{definition}
\end{samepage}
\begin{corollary}\label{cor:clisti}
  In the abstract phase, $\PListI{c}[\F{\X}] = \ListC[\F{\X}] = \F[\ListV{\X}]$.
\end{corollary}

\begin{example}\label{ex:reverse}
  Using the induction principle for the type of lists with linear credits, we may define a linear-time list reverse algorithm that uses the stored credits as follows.
  \begin{center}
    \begin{minipage}[t]{0.55\textwidth}
      \iblock{
        \mrow{\DEFN{revAppend} : \LolliV{\PListI{1}[\F{X}]}{\PowerC{\ListV{\X}}{\F[\ListV{\X}]}}}
        \mrow{\DEFN{revAppend}~\ListCNil \isdef \plam{acc}{\ret{acc}}}
        \mhang{\DEFN{revAppend}~(\ListCCons{a}{l}) \isdef \plam{acc}{}}{
          \mrow{\bind{\CreditSpend{1}[a]}{x}}
          \mrow{\DEFN{revAppend}~l~(\ListVCons{x}{acc})}
        }
      }
    \end{minipage}%
    \begin{minipage}[t]{0.45\textwidth}
      \iblock{
        \mrow{\DEFN{reverse} : \LolliV{\PListI{1}[\F{\X}]}{\F[\ListV{\X}]}}
        \mrow{\DEFN{reverse}~l \isdef \DEFN{revAppend}~l~\ListVNil}
      }
    \end{minipage}%
  \end{center}
  Note that the only cost effect occurs within $\CreditSpend{1}$;
  this means that abstractly, both $\DEFN{revAppend}$ and $\DEFN{reverse}$ are zero-cost, because the credits required are pre-paid into the input list.
\end{example}

Using this credit-assuming reverse function, we may implement batched queues using the banker's view, as the potential function assigned one unit of potential per element of the back list.

\begin{example}\label{ex:banker-bq}
  Let $\B \isdef \CopowerC{\LV}{\PListI{1}[\F{\NatV}]}$ describe pairs of lists of natural numbers where each element of the second list is equipped with one credit.
  To define $\DEFN{enqueue} : \ExpV{\NatV}{\LolliV{\B}{\B}}$, we first define an auxiliary function to accept a credit, to be stored alongside the new list element:
  \begin{align*}
    &\DEFN{enqueue'} : \LolliV{\Credit{1}[\CopowerC{\NatV}{\B}]}{\B}
  \end{align*}
  Then, $\DEFN{enqueue} \isdef \vlam{n}{\clam{b}{\DEFN{enqueue'}~(\CreditSave{1}[n, b])}}$, using the $\CreditSave{1}$ derived form to pre-pay for the credit, revealed as amortized cost in the abstract phase.
  The other operations are similar, where $\DEFN{dequeue}$ makes use of $\DEFN{reverse}$ from \cref{ex:reverse}.
\end{example}

The type $\B$ does not meet the queue interface of \cref{ex:ephemeral-queue}, as it does not perform abstraction to masquerade as a single list.
Following \cref{ex:sequential-bq}, it is possible to induce this abstraction using a function $\alpha : \LolliV{\B}{\F[\LV]}$.
Alternatively, we may use a \emph{phased quotient}~\citep[\S 2.3]{grodin-li-harper>2026} to abstract implicitly while maintaining the locality of the banker's method.

\begin{example}
  The representation type $\B$ of \cref{ex:banker-bq} may be augmented to be suitably abstract as an implementation of queues by applying a quotient in the abstract phase:
  \par\iblock{
    \mhang{\KW{data}~\A : \CC~\KW{where}}{
      \mrow{\LABEL{inj} : \LolliV{\B}{\A}}
      \mrow{\LABEL{tilt} : \ExpV{\AbsV}{\DProdV{x : \NatV}{\DProdV{l_1~l_2 : \ListV{\NatV}}{\LABEL{inj}~(\ListVAppend{l_1}{\ListVSing{x}}, \ret{l_2}) = \LABEL{inj}~(l_1, \ret[\ListVAppend{l_2}{\ListVSing{x}}])}}}}
    }
  }\noindent
  Implicitly, we use \cref{cor:clisti}.
  This type $\A$ is thus equivalent in the abstract phase to $\F[\LV]$.
  When the operations sketched in \cref{ex:banker-bq} are shown to preserve this abstract quotient, this structure implements the ephemeral queue interfaces of \cref{ex:ephemeral-queue,ex:ephemeral-queue-cost}.
\end{example}

Generalizing linear-credit lists of the previous section, we demonstrate the case of lists carrying linear and triangular (\ie, $c_1n + c_2\binom{n}{2}$) credits.

\begin{samepage}
\begin{definition}[{\AgdaFormalized{Calf.Computation.PList2}}]\label{lem:plistii}
  Define $\PListII{c_1}{c_2}{\A}$ to be the following inductive type family:
  \par\iblock{
    \mhang{\KW{data}~\PListII{c_1}{c_2}{\A} : \CC~\KW{where}}{
      \mrow{\ListCNil : \LolliV{\TopC}{\PListII{c_1}{c_2}{\A}}}
      \mrow{\ListCCons{\_}{\_} : \LolliV{\TensorC{\Credit{c_1}{\A}}{\PListII{c_2 + c_1}{c_2}{\A}}}{\PListII{c_1}{c_2}{\A}}}
    }
  }\noindent
  Note that this inductive family varies $c_1$ and is thus \emph{not} representable via $\ListC[-]$.
\end{definition}
\end{samepage}

\begin{corollary}\label{cor:clistii}
  In the abstract phase, $\PListII{c_1}{c_2}[\F{\X}] = \ListC[\F{\X}] = \F[\ListV{\X}]$.
\end{corollary}

Using this variety of credit-carrying list, it is possible to implement pre-paid versions of quadratic-cost algorithms, such as insertion sort.
We return to this line of development in \cref{sec:giralf}.

\subsection{Credit-Carrying Trees}

\NewDocumentCommand{\RBTreeC}{m m m}{\CTY{RBTree}~{#1}~{#2}~{#3}}

\begin{example}
  \citet{tarjan>1985} established an amortized analysis of red-black trees~\citep{guibas-sedgewick>1978} in which each black node contains credits computed based on the color of its child nodes.
  We may represent this construction in the banker's view as follows, where $0 \le c(x_1, x_2) \le 2$ is the number of credits to be stored at a black node with child nodes of colors $x_1$ and $x_2$.
\begin{samepage}
  \par\iblock{
    \mhang{\KW{data}~\RBTreeC{(\A : \CC)}{(x : \VTY{Color})}{(h : \NatV)} : \CC~\KW{where}}{
      \mrow{\LABEL{empty} : \LolliV{\TopC}{\RBTreeC{\A}{\LABEL{black}}{0}}}
      \mrow{\LABEL{red} : \LolliV{\TensorC{\RBTreeC{\A}{\LABEL{black}}{h}}{\TensorC{\A}{\RBTreeC{\A}{\LABEL{black}}{h}}}}{\RBTreeC{\A}{\LABEL{red}}{h}}}
      \mrow{\LABEL{black} : \LolliV{\Credit{c(x_1, x_2)}[\TensorC{\RBTreeC{\A}{x_1}{h}}{\TensorC{\A}{\RBTreeC{\A}{x_2}{h}}}]}{\RBTreeC{\A}{\LABEL{black}}{(1 + h)}}}
    }
  }
\end{samepage}\noindent
  This construction can be adapted to support the appropriate abstraction using the techniques used on batched queues, following \citet{grodin-li-harper>2026}.
\end{example}

\NewDocumentCommand{\SplayTreeC}{m m}{\CTY{SplayTree}~{#1}~{#2}}

\begin{example}
  The amortized splay tree~\citep{sleator-tarjan>1985} data structure stores $\ceil{\lg n}$ credits at each node whose subtree is of size $n$.
  Such credits can be represented in the following inductive type, where $c(n_1, n_2) \isdef \ceil{\lg(n_1 + 1 + n_2)}$.
\begin{samepage}
  \par\iblock{
    \mhang{\KW{data}~\SplayTreeC{(\A : \CC)}{(n : \NatV)} : \CC~\KW{where}}{
      \mrow{\LABEL{leaf} : \LolliV{\TopC}{\SplayTreeC{\A}{0}}}
      \mrow{\LABEL{node} : \LolliV{\Credit{c(n_1, n_2)}[\TensorC{\SplayTreeC{\A}{n_1}}{\TensorC{\A}{\SplayTreeC{\A}{n_2}}}]}{\SplayTreeC{\A}{(n_1 + 1 + n_2)}}}
    }
  }
\end{samepage}\noindent
  The credits associated to splay trees are notoriously difficult to annotate (and analyze automatically) in AARA~\citep{hofmann-leutgeb-obwaller-moser-zuleger>2022}.
  However, in the dependent setting of the present work, it is straightforward to include credits just as originally described by \citeauthor{sleator-tarjan>1985}.
\end{example}

\section{Giralf: A Graded, Inferential, Resource-Aware Logical Framework}\label{sec:giralf}

In order to streamline the development of programs involving credits and debits, we now define a graded substructural type theory called \emph{Giralf} overlaid upon the existing dependent type theory of Calf; drawing inspiration from AARA~\citep{hofmann-jost>2003,hoffmann-jost>2022}, the credit and debit type operators internalize the graded judgmental structure of Giralf.
The types of Giralf are taken from Calf, and thus Giralf programs semantically constitute a well-behaved subclass of Calf programs, demonstrating that AARA-like programs exist as a sub-language of the computations of Calf.
Furthermore, by adapting the techniques of AARA, programs written in the Giralf language are amenable to a form of automated cost inference; thus, as a semantic sub-language of Calf, Giralf facilitates the integration of manual and automatic cost verification.

\subsection{A Graded Syntax (\AgdaFormalized{Calf.Giralf})}\label{sec:giralf:syntax}

The syntax of Giralf is similar to that of \citet{das-balzer-hoffmann-pfenning-santurkar>2021}, but with two major differences:
\begin{enumerate}
  \item \emph{Dependency.}
    Giralf admits dependency of substructural resources on structural values, like in the dependent linear/non-linear type theory of \citet{krishnaswami-pradic-benton>2015}.
    For example, credit-carrying types refer to structural values of type $\Cost$.
    This is important for handling indexed inductive types, such as the quadratic-credit lists of \cref{sec:banker:plisti}.

  \item \emph{Recursion.}
        It is typical for AARA-like languages to admit unbounded recursion.
        In Giralf, we only consider the structural recursion principles induced by inductive types.%
        \footnote{It is possible to treat unbounded recursion in Calf as an effect \citep{niu-harper>2022}; however, it is not clear that this effect is compatible with the fracture and gluing principle central to this work.}

\end{enumerate}

\NewDocumentCommand{\SplitZ}{m}{\strut}
\NewDocumentCommand{\SplitII}{m m m}{{#1} \ge {#2} + {#3}}

\NewDocumentCommand{\CtxEmpty}{}{\cdot}

\NewDocumentCommand{\GammaV}{}{\VAL*{\Gamma}}
\NewDocumentCommand{\DeltaC}{o}{\CMP*{\Delta\IfValueT{#1}{_{#1}}}}
\NewDocumentCommand{\Giralf}{o m m m m}{\IfValueT{#1}{{#1} \mid} {#2} \vdash^{#3} {#4} : {#5}}

We now define Giralf, whose types are inherited from Calf (including those of \cref{sec:banker}).
The typing judgment $\Giralf[\GammaV]{\DeltaC}{q}{e}{\A}$ means that given a structural context $\GammaV$, a linear context $\DeltaC = \A[1], \cdots, \A[n]$ dependent on $\GammaV$, and additional credits $q : \Cost$ also dependent on $\GammaV$, the program $e$ has type $\A$.
For readability, we leave $\GammaV$ implicit, only notating the extension beyond the ambient $\GammaV$.

\NewDocumentCommand{\LetG}{m m m}{\KW{let}\, {#2} = {#1} \, \KW{in}\, {#3}}
\NewDocumentCommand{\ChargeG}{m o}{\KW{spend}\langle{#1}\rangle\IfValueT{#2}{({#2})}}

We first define the rules that apply uniformly over types.
The variable (identity) and let-binding (cut) rules are standard from graded type theory:
\begin{mathpar}
  \infer
  {\SplitZ{q}}
  {\Giralf{a : \A}{q}{a}{\A}}

  \infer
  {
    \SplitII{q}{q_1}{q_2} \\
    \Giralf{\DeltaC[1]}{q_1}{e_1}{\A} \\
    \Giralf{\DeltaC[2], a : \A}{q_2}{e_2}{\B}
  }
  {\Giralf{\DeltaC[1], \DeltaC[2]}{q}{\LetG{e_1}{a}{e_2}}{\B}}
\end{mathpar}
The rule for spending credits follows AARA~\citep{hofmann-jost>2003,hoffmann-jost>2022,das-balzer-hoffmann-pfenning-santurkar>2021}.
Justified by having sufficient credits in the context, this construct incurs cost, to be annotated on programs as a cost model.
\begin{mathpar}
  \infer
  {
    \SplitII{q}{p}{q'} \\
    \Giralf{\DeltaC}{q'}{e}{\A}
  }
  {\Giralf{\DeltaC}{q}{\ChargeG{p}[e]}{\A}}
\end{mathpar}

\subsubsection{Standard Linear Types}

\NewDocumentCommand{\InjG}{m o}{\KW{inj}_{#1}~\IfValueT{#2}{({#2})}}
\NewDocumentCommand{\ProjG}{m m}{\KW{proj}_{#1}~{#2}}
\NewDocumentCommand{\CaseG}{m m m m m}{\KW{case}({#1};{#2}.{#3};{#4}.{#5})}

The rules for standard linear types are as usual.
Of note, Giralf includes negative types, such as lazy products (like $\lambda$-amor \citep{rajani-gaboardi-garg-hoffmann>2021} but unlike AARA):
\begin{mathpar}
  \infer
  {
    \Giralf{\DeltaC}{q}{e_1}{\A[1]} \\
    \Giralf{\DeltaC}{q}{e_2}{\A[2]}
  }
  {\Giralf{\DeltaC}{q}{(e_1, e_2)}{\ProdC{\A[1]}{\A[2]}}}

  \infer
  {\Giralf{\DeltaC}{q}{e}{\ProdC{\A[1]}{\A[2]}}}
  {\Giralf{\DeltaC}{q}{\ProjG{i}{e}}{\A[i]}}
\end{mathpar}
We now turn our attention to types that interact with the credit context.

\subsubsection{Credit and Debit}\label{sec:giralf:syntax:credit}

\NewDocumentCommand{\StoreG}{o m}{\KW{store}\IfValueT{#1}{\langle {#1} \rangle} ({#2})}
\NewDocumentCommand{\ReleaseG}{m m m}{\KW{let}\, \StoreG{#2} = {#1} \, \KW{in}\, {#3}}
\NewDocumentCommand{\GetG}{m m}{\KW{get}\langle {#1} \rangle ({#2})}
\NewDocumentCommand{\PayG}{m}{\KW{pay}({#1})}

The rules for the credit operator are analogous to those given by \citet{das-balzer-hoffmann-pfenning-santurkar>2021} and \citet{rajani-gaboardi-garg-hoffmann>2021}, internalizing credits from the credit context as a type former.
\begin{mathpar}
  \infer
  {
    \SplitII{q}{p}{q'} \\
    \Giralf{\DeltaC}{q'}{e}{\A}
  }
  {\Giralf{\DeltaC}{q}{\StoreG[p]{e}}{\Credit{p}{\A}}}

  \infer
  {
    \SplitII{q}{q_1}{q_2} \\
    \Giralf{\DeltaC[1]}{q_1}{e_1}{\Credit{p}{\A}} \\
    \Giralf{\DeltaC[2], a : \A}{p + q_2}{e_2}{\B}
  }
  {\Giralf{\DeltaC[1], \DeltaC[2]}{q}{\ReleaseG{e_1}{a}{e_2}}{\B}}
\end{mathpar}
The rules for the debit operator are analogous to those given by \citet{das-balzer-hoffmann-pfenning-santurkar>2021}.
\begin{mathpar}
  \infer
  {
    \SplitII{q'}{p}{q} \\
    \Giralf{\DeltaC}{q'}{e}{\A}
  }
  {\Giralf{\DeltaC}{q}{\GetG{p}{e}}{\Debit{p}{\A}}}

  \infer
  {
    \SplitII{q}{p}{q'} \\
    \Giralf{\DeltaC}{q'}{e}{\Debit{p}{\A}}
  }
  {\Giralf{\DeltaC}{q}{\PayG{e}}{\A}}
\end{mathpar}

\subsubsection{Linear-Credit Lists}\label{sec:giralf:syntax:plist}

\NewDocumentCommand{\NilG}{}{\ListVNil}
\NewDocumentCommand{\ConsG}{m m}{\ListVCons{#1}{#2}}
\NewDocumentCommand{\FoldrIG}{m m m m m}{\KW{foldr}[{#1}; {#2}.{#3}.{#4}]({#5})}
\NewDocumentCommand{\ParaIG}{m m m m m}{\KW{para}[{#1}; {#2}.{#3}.{#4}]{#5}}
\NewDocumentCommand{\FoldrIIG}{m m m m m m m}{\KW{foldr}\{{#1}.{#2}\}[{#1}.{#3}; {#1}.{#4}.{#5}.{#6}]({#7})}

The AARA-like lists with linear credit described in \cref{sec:banker:plisti} can be smoothly incorporated into Giralf, providing an elimination form in terms of structural recursion:
\begin{mathpar}
  \infer
  {\SplitZ{q}}
  {\Giralf{\CtxEmpty}{q}{\NilG}{\PListI{p}{\A}}}

  \infer
  {
    \SplitII{q}{p + q_1}{q_2} \\
    \Giralf{\DeltaC[1]}{q_1}{a}{\A} \\
    \Giralf{\DeltaC[2]}{q_2}{l}{\PListI{p}{\A}}
  }
  {\Giralf{\DeltaC[1], \DeltaC[2]}{q}{\ConsG{a}{l}}{\PListI{p}{\A}}}

  \infer
  {
    \Giralf{\CtxEmpty}{0}{e_0}{\B} \\
    \Giralf{a : \A, b : \B}{p}{e_1}{\B} \\
    \Giralf{\DeltaC}{q}{e}{\PListI{p}{\A}}
  }
  {\Giralf{\DeltaC}{q}{\FoldrIG{e_0}{a}{b}{e_1}{e}}{\B}}
\end{mathpar}

Using this recursion principle, it is possible to implement a variety of linear-time algorithms whose cost is pre-paid for by the available credits.

\begin{lemma}\label{lem:paramorphism}
  The following credit-aware \emph{paramorphism}~\citep{meertens>1992} is derivable:
  \begin{mathpar}
    \infer
      {
        \SplitII{q}{q_1}{q_2} \\
        \Giralf{\CtxEmpty}{q_2}{e_0}{\B} \\
        \Giralf{a : \A, b : \ProdC{\Debit{q_2}{\B}}{\PListI{p}{\A}}}{p + q_2}{e_1}{\B} \\
        \Giralf{\DeltaC}{q_1}{e}{\PListI{p}{\A}}
      }
      {\Giralf{\DeltaC}{q}{\ParaIG{e_0}{a}{b}{e_1}{(e)}}{\B}}
  \end{mathpar}
  This construction implicitly threads through credits available at the top level using the debit operator; moreover, using lazy products, it offers the choice between the recursive result and the current sublist.
\end{lemma}

\NewDocumentCommand{\BoolVIf}{m m m}{\KW{if}~{#1}~\KW{then}~{#2}~\KW{else}~{#3}}

\begin{example}[\AgdaFormalized{Examples.Giralf.InsertionSort}]\label{ex:insert}
  The main subroutine of insertion sort is definable in Giralf as a term
  \[ \Giralf[p : \Cost, x : \NatV]{l : \PListI{1 + p}[\F{\NatV}]}{p}{\DEFN{insert}~x~l}{\PListI{p}[\F{\NatV}]} \]
  using a paramorphism $\DEFN{insert}~x~l \isdef \ParaIG{\ConsG{\ret{x}}{\NilG}}{(\ret{y})}{b}{e_1}{}$ where
  \[ e_1 \isdef \ChargeG{1}[\BoolVIf{x \le y}{(\ConsG{\ret{x}}{\ConsG{\ret{y}}{\ProjG{2}{b}}})}{(\ConsG{\ret{y}}{\PayG{\ProjG{1}{b}}})}]. \]
  In the base case, a singleton list containing only $x$ is created.
  In the inductive case, if the new element $x$ is smaller than the head element $y$, then $x$ and $y$ are placed on the front of the original list, recovered via $\ProjG{2}{b}$; if $x$ is larger than $y$, then $y$ is placed on the front of the recursive call, passing down the credits that will eventually be attached to the new list node.
\end{example}

\subsubsection{Quadratic-Credit Lists}\label{sec:giralf:syntax:plistii}
The introduction rules for quadratic-credit lists are standard, following \citet{hoffmann-hofmann>2010,hoffmann-hofmann>2010-aplas}:
\begin{mathpar}
  \infer
  {\SplitZ{q}}
  {\Giralf{\CtxEmpty}{q}{\NilG}{\PListII{p_1}{p_2}{\A}}}

  \infer
  {
    \SplitII{q}{p_1 + q_1}{q_2} \\
    \Giralf{\DeltaC[1]}{q_1}{a}{\A} \\
    \Giralf{\DeltaC[2]}{q_2}{l}{\PListII{p_2 + p_1}{p_2}{\A}}
  }
  {\Giralf{\DeltaC}{q}{\ConsG{a}{l}}{\PListII{p_1}{p_2}{\A}}}
\end{mathpar}
The elimination form, however, diverges from existing presentations due to the requirement here to use only bounded recursion.
In particular, due to the status of quadratic-credit lists as an inductive family, their recursion principle makes nontrivial use of the structural context and dependency:
\begin{mathpar}
  \infer
  {
    \Giralf[r : \Cost]{\CtxEmpty}{0}{e_0}{\B(r)} \\
    \Giralf[r : \Cost]{a : \A, b : \B(p_2 + r)}{r}{e_1}{\B(r)} \\
    \Giralf{\DeltaC}{q}{e}{\PListII{p_1}{p_2}{\A}}
  }
  {\Giralf{\DeltaC}{q}{\FoldrIIG{r}{\B(r)}{e_0}{a}{b}{e_1}{e}}{\B(p_1)}}
\end{mathpar}
This rule eliminates into a \emph{family} of types $\B : \ExpV{\Cost}{\CC}$ indexed by the linear credit coefficient; the family of outputs is required precisely because $\PListII{-}{p_2}{\A} : \ExpV{\Cost}{\CC}$ is not an inductively-defined type, but rather an inductively-defined \emph{family}.
At the top level, the result type is $\B(p_1)$, where $p_1$ is the linear coefficient of the list being eliminated. Inductively, both the base case and the inductive case must construct an element of type $\B(r)$, where $r : \Cost$ is a freshly bound variable, required because the amount of linear credit changes inductively.
This rule is in place of an ad-hoc rule for resource-polymorphic recursion \citep{hoffmann-hofmann>2010-aplas}.

\begin{example}[\AgdaFormalized{Examples.Giralf.InsertionSort}]\label{ex:isort}
  The insertion sort algorithm, which in the worst case costs $\binom{n}{2}$ on a list of length $n$ when counting comparisons, can be implemented as a term
  \[ \Giralf{l : \PListII{0}{1}[\F{\NatV}]}{0}{\DEFN{isort}~l}{\PListI{0}[\F{\NatV}]} \]
  by iterating the $\DEFN{insert}$ algorithm of \cref{ex:insert}:
  \[ \DEFN{isort} \isdef \FoldrIIG{r}{\PListI{r}[\F{\NatV}]}{\NilG}{(\ret{x})}{l}{~\DEFN{insert}~x~l}{}. \]
  The family $\B(r) \isdef \PListI{r}[\F{\NatV}]$ relies on the specification of $\DEFN{insert}$, which is parametric in the linear credits ($p$ in \cref{ex:insert}).
  The fact that this program only takes a triangular credit annotation of $1$ on the input list guarantees implicitly that insertion sort has the desired cost upper bound.
\end{example}

\subsection{A Resource-Aware Semantics}\label{sec:giralf:semantics}

Semantically, Giralf can be defined as a \emph{sub-language} of Calf.
Every Giralf type is already present in Calf, and
a program typing judgment
$\Giralf{\DeltaC}{q}{e}{\A}$
in Giralf is interpreted as a Calf program
$e : \LolliV*{\Credit{q}[\CMP*{\otimes}\DeltaC]}{\A}$,
moving credits via \cref{lem:credit-functorial,lem:credit-monoidal,lem:credit-debit-adjoint}.
As an invariant, this interpretation has zero cost in the abstract phase.

\begin{example}[{\AgdaFormalized{Calf.Giralf}}]\label{sem:spend}
  The term $\StoreG[p]{e}$ is interpreted as the composite
\[\begin{tikzcd}
  {\Credit{q}[\CMP*{\otimes}\DeltaC]} && {\Credit{p + q'}[\CMP*{\otimes}\DeltaC]} && {\Credit{p}[\Credit{q'}[\CMP*{\otimes}\DeltaC]]} & {\Credit{p}{\A}}
  \arrow[lmmv, "{\text{\cref{lem:credit-functorial}}}", from=1-1, to=1-3]
  \arrow[mmv, "{\text{\cref{lem:credit-monoidal}}}", r,-,double equal sign distance,double, from=1-3, to=1-5]
  \arrow[lmmv, "{\Credit{p}[e]}"{inner sep=.8ex}, from=1-5, to=1-6]
\end{tikzcd}\]
  credits as output.
  Then, the term $\ChargeG{p}[e]$ is interpreted as post-composition of the above chain with $\CreditSpend{p} : \LolliV{\Credit{p}{\A}}{\A}$.
  The use of $\CreditSpend{p}$, \emph{not} $\charge{p}$, highlights a fundamental difference between Calf and of Giralf: although the cost effect $\charge{p} : \LolliV{\A}{\A}$ may be performed arbitrarily in Calf, cost may only be incurred in Giralf via the operation $\CreditSpend{p} : \LolliV{\Credit{p}{\A}}{\A}$ when the cost has already been paid for upfront.
\end{example}

\begin{remark}\label{rem:prepay}
  Given a Giralf program interpreted as $e : \LolliV*{\Credit{q}[\CMP*{\otimes}\DeltaC]}{\A}$,
  pre-paying for the $q$ credits via the Calf program
  $e \circ \CreditSave{q} : \LolliV*{\CMP*{\otimes}\DeltaC}{\A}$
  reframes the credit context $q$ as the amortized cost publicized in the abstract phase rather than as credits attached to the input.
\end{remark}

The Giralf language presented here is directly inspired by AARA~\citep{hoffmann-hofmann>2010} and, when considering terminating functions, generalizes it.
In AARA, adapting to the notation of this work, it is common to consider only positive%
\footnote{Due to the call-by-value nature of AARA, the function types of AARA correspond to the type $\F[\LolliV{\A}{\B}]$.}
types, defined by a grammar such as
\[ \A,\B,\C \coloncolonequals \F{\X} \mid \TopC \mid \TensorC{\A}{\B} \mid \SumC{\A}{\B} \mid \Credit{p}{\A} \mid \PListI{p}{\A} \mid \PListII{p_1}{p_2}{\A}. \]
Note that this is an inductively-defined subset of the computation types available in Calf (and thus Giralf).
In the semantics of AARA~\citep{hoffmann-jost>2022}, each of these types is assigned a set of values $\AARAValues{\A} : \VV$, and then a potential function $\Phi_{\A}$ is defined by induction on types in \cref{fig:aara-semantics}.
\begin{figure}
  \begin{align*}
    \AARAValues{\A} &: \VV &
    \Phi_{\A} &: \ExpV{\AARAValues{\A}}{\Cost} \\
    \AARAValues{\F{\X}} &\isdef \X &
    \Phi_{\F{\X}}(x) &\isdef 0 \\
    \AARAValues{\TopC} &\isdef \UnitV &
    \Phi_{\TopC}() &\isdef 0 \\
    \AARAValues{\TensorC{\A}{\B}} &\isdef \ProdV{\AARAValues{\A}}{\AARAValues{\B}} &
    \Phi_{\TensorC{\A}{\B}}(a, b) &\isdef \Phi_{\A}(a) + \Phi_{\B}(b) \\
    \AARAValues{\SumC{\A[1]}{\A[2]}} &\isdef \SumV{\AARAValues{\A[1]}}{\AARAValues{\A[2]}} &
    \Phi_{\SumC{\A[1]}{\A[2]}}(\InjG{i}{a_i}) &\isdef \Phi_{\A[i]}(a_i) \\
    \AARAValues{\Credit{p}{\A}} &\isdef \AARAValues{\A} &
    \Phi_{\Credit{p}{\A}}(a) &\isdef p + \Phi_{\A}(a) \\
    \AARAValues{\PListI{p}{\A}} &\isdef \ListV{\AARAValues{\A}} &
    \Phi_{\PListI{p}{\A}}(l) &\isdef \ListVLength{l} \cdot p + \sum_{a \in l} \Phi_{\A}(a) \\
    \AARAValues{\PListII{p_1}{p_2}{\A}} &\isdef \ListV{\AARAValues{\A}} &
    \Phi_{\PListII{p_1}{p_2}{\A}}(l) &\isdef \ListVLength{l} \cdot p_1 + \binom{\ListVLength{l}}{2} \cdot p_2 + \sum_{a \in l} \Phi_{\A}(a)
  \end{align*}
  \caption{The potential-based semantics of AARA types~\citep{hoffmann-hofmann>2010}.}\label{fig:aara-semantics}
\end{figure}

\begin{remark}
  In such a semantics, it is not obvious how to incorporate negative types.
  For example, what potential (of type $\Cost$) would a value of type $\ProdC{\Credit{1}{\TopC}}{\Credit{2}{\TopC}}$ have---or worse, $\DPowerC{n : \NatV}{\Credit{n}{\TopC}}$?
\end{remark}

We may build such a potential function $\Phi_{\A}$ into a type, $\PotentialC{\AARAValues{\A}}{\Phi_{\A}}$, using \cref{def:PotentialC}.
In fact, the potential function $\Phi_{\A}$ is already contained within the Calf type $\A$.

\begin{theorem}\label{thm:calf-aara-types}
  Let $\A$ be a type in the AARA grammar given above.
  Viewed as a Calf type, $\A$ is equivalent to the type $\PotentialC{\AARAValues{\A}}{\Phi_{\A}}$.
\end{theorem}
\begin{proof}[Proof Sketch]
  By induction on the AARA type grammar.
  From the physicist's view, as the type $\PotentialC{\AARAValues{\A}}{\Phi_{\A}}$ is given in terms of $\F{\AARAValues{\A}}$ (by \cref{def:PotentialC}), the cases follow by the fact that $\F{}$ (as a left adjoint) preserves positive types.
  Or, from the banker's view, the cases follow from \cref{lem:PotentialC,lem:credit-monoidal} and credits commuting with positive types.
\end{proof}
This shows not only that the potential functions $\Phi_{\A}$ of AARA are present in Calf types, but also that the connectives used here generalize AARA in a compatible way.
Then, framed in terms of $\Phi_{\A}$, the soundness theorem for AARA follows immediately from the Kripke semantics of Calf.
\begin{theorem}[Soundness of AARA]
  Let $\Giralf{\DeltaC}{q}{e}{\A}$ be a typing judgment in AARA, and suppose the program $e$ incurs $p : \Cost$ cost upon evaluation in some environment $\delta : \AARAValues{\DeltaC}$.
  Then, the specification $q$ is a sound upper bound for $p$, up to the following conservation condition:
  \[ p + \Phi_{\A}(e(\delta)) \le \Phi_{\CMP*{\otimes}\DeltaC}(\delta) + q. \]
\end{theorem}
\begin{proof}
  View the AARA program $e$ as a Giralf program
  $e : \LolliV*{\CMP*{\otimes}\DeltaC}{\A}$.
  By \cref{thm:calf-aara-types}, we have
  \begin{align*}
    \DeltaC = \PotentialC{\AARAValues{\CMP*{\otimes}\DeltaC}}{\Phi_{\CMP*{\otimes}\DeltaC}}  &&\text{ and }  && \A = \PotentialC{\AARAValues{\A}}{\Phi_{\A}}.
  \end{align*}
  Then, in the Kripke semantics of Calf, $e$ is interpreted as the following lax commutative square:
\[\begin{tikzcd}
  {\Interp{\F{\AARAValues{\CMP*{\otimes}\DeltaC}}}} &&& {\Interp{\F{\AARAValues{\A}}}} \\
  {\Interp{\F{\AARAValues{\CMP*{\otimes}\DeltaC}}}} &&& {\Interp{\F{\AARAValues{\A}}}}
  \arrow[mm, "{\Interp[\top]{e}}", from=1-1, to=1-4]
  \arrow[mm, ""{name=0, anchor=center, inner sep=0}, "{\varphi(\Phi_{\CMP*{\otimes}\DeltaC})}"', from=1-1, to=2-1]
  \arrow[mm, ""{name=1, anchor=center, inner sep=0}, "{\varphi(\Phi_{\A}(a))}", from=1-4, to=2-4]
  \arrow[mm, "{\Interp[\ABS]{e} \;\circ\; \charge{q}}"', from=2-1, to=2-4]
  \arrow["\ge"{description}, draw=none, from=1, to=0]
\end{tikzcd}\]
  The cost portion of this square is precisely the desired conservation condition.
\end{proof}
Thus, Calf can be viewed as a conservative extension of AARA.
When restricting attention to the AARA-like types in Giralf, the language behaves just like AARA, but Giralf supports additional types, data abstraction, and manual verification.
Conversely, when Calf programs happen to lie in the Giralf sub-language, it is possible to automatically infer cost bounds.

\subsection{An Inference Algorithm}\label{sec:giralf:inference}

\NewDocumentCommand{\Inf}{m}{\underline{#1}}

In Giralf, credits are placed within data structures to ensure the availability of a credit whenever cost is incurred.
Beyond streamlining the manual development of programs involving credits, this realization of the banker's view
enables automated \emph{cost inference} by linear programming as in AARA and RaML~\citep{hoffmann-aehlig-hofmann>2012-raml}.
Cost inference takes a Calf program and finds an analogous program in Giralf which, by construction, assumes all cost is prepaid for upfront.

\begin{definition}[Cost Inference]
  Let $e : \LolliV{\A}{\B}$ be a Calf program.
  An \emph{inference for $e$} is a Giralf term $\Giralf{a : \Inf{\A}}{q}{\Inf{e}}{\Inf{\B}}$ such that
  \begin{enumerate}
    \item assuming $\NAbsV$, it is the case that $\Inf{\A} = \A$, $\Inf{\B} = \B$, and $\Inf{e} = e$; and
    \item assuming $\AbsV$, it is the case that $\Inf{e}$ is the cost erasure of $e$.
  \end{enumerate}
\end{definition}

The assumption of $\NAbsV$ isolates only the true execution behavior of the program;
in particular, this assumption erases credit annotations.

\begin{lemma}
  Assuming $\NAbsV$, it is the case that $\GlueeC{\A[\top]}{\A[\ABS]}{\alpha} = \A[\top]$.
\end{lemma}
\begin{samepage}
  \begin{corollary}\label{cor:nabs}
    Assuming $\NAbsV$, it is the case that:
    \begin{enumerate}
      \item $\Credit{p}{\A} = \A$ and $\ChargeG{c} = \charge{c}$;
      \item $\Debit{p}{\A} = \A$; and
      \item $\PListI{c}{\A} = \PListII{c_1}{c_2}{\A} = \ListC{\A}$.
    \end{enumerate}
  \end{corollary}
\end{samepage}

\begin{example}
  The Giralf insertion sort program of \cref{ex:isort}, annotated with $\ChargeG{-}$, is an inference for the Calf insertion sort program of \citet{niu-sterling-grodin-harper>2022}, annotated with $\charge{-}$.
\end{example}

\NewDocumentCommand{\Q}{}{{?}}
\NewDocumentCommand{\VQ}{}{{\VAL*{?}}}

An \emph{inference algorithm} is licensed to alter the types in a program to include credits.
Note that inference may fail, as the costs included in an arbitrary Calf program can be arbitrarily complex~\citep{niu-sterling-grodin-harper>2022}.
Building on the linear programming-based cost inference techniques of AARA~\citep{hofmann-jost>2003}, we now describe a cost inference algorithm for a sub-language of Calf.

\subsubsection{Skeletal Translation}

For inference, we restrict attention to types in the following grammar:
\[
  \A,\B,\C \coloncolonequals \TopC \mid \SumC{\A}{\B} \mid \ProdC{\A}{\B} \mid \ListC[\F{\X}].
\]
We choose these as representative cases, but it is straightforward to accommodate similar types of AARA.
Inference for the simple positive types $\TopC$ and $\SumC{\A}{\B}$ is well-understood~\citep{hofmann-jost>2003}; inference for the lazy product type $\ProdC{\A}{\B}$ is novel; and inference for the list type $\ListC[\F{\X}]$ is involved when quadratic credits are considered due to the requirement of structural recursion.

\NewDocumentCommand{\TransTy}{m}{\Inf{#1}} %
\NewDocumentCommand{\TransExp}{m}{\TransTy{#1}}

Inference begins by inductively defining $\TransTy{\A} : \CC$, an augmentation of $\A$ with the structure of credits and debits that leaves the precise numbers yet unspecified (indicated by a $\Q$ symbol).
\begin{align*}
  \TransTy{\A}                   & : \CC                                                                    \\
  \TransTy{\TopC}                & \isdef \TopC                                                             \\
  \TransTy{\SumC{\A[1]}{\A[2]}}  & \isdef \SumC{\Credit{\Q}{\TransTy{\A[1]}}}{\Credit{\Q}{\TransTy{\A[2]}}} \\
  \TransTy{\ProdC{\A[1]}{\A[2]}} & \isdef \ProdC{\Debit{\Q}{\TransTy{\A[1]}}}{\Debit{\Q}{\TransTy{\A[2]}}}  \\
  \TransTy{\ListC[\F{\X}]}       & \isdef \PListII{\Q}{\Q}[\F{\X}]
\end{align*}
Note that $\ExpV{\NAbsV}{(\TransTy{\A} = \A)}$ by \cref{cor:nabs}.
The duality of sums and products appears via the duality of credits and debits: when eliminating from a sum/introducing a product, each case/component may require a different amount of credits.
We annotate lists with linear and quadratic credits.

On terms, inference augments Calf programs $e : \LolliV{\A}{\B}$ to Giralf programs $\Inf{e}$ such that
$\Giralf{\TransTy{\A}}{\Q}{\TransExp{e}}{\TransTy{\B}}$.
For example, the cost effect is reframed as spending credits, and the sum and product cases follow routinely from the type translations.
\begin{align*}
  \TransExp{\charge{c}[e]}                 & \isdef \ChargeG{c}[\TransExp{e}]                                                                         \\
  \TransExp{\InjG{i}{e}}                   & \isdef \InjG{i}{(\StoreG[\Q]{\TransExp{e}})}                                                             \\
  \TransExp{\CaseG{e}{a_1}{e_1}{a_2}{e_2}} & \isdef \CaseG{e}{a_1'}{\ReleaseG{a_1'}{a_1}{\TransExp{e_1}}}{a_2'}{\ReleaseG{a_2'}{a_2}{\TransExp{e_2}}} \\
  \TransExp{(e_1, e_2)}                    & \isdef (\GetG{\Q}{\TransExp{e_1}}, \GetG{\Q}{\TransExp{e_2}})                                            \\
  \TransExp{\ProjG{i}{e}}                  & \isdef \PayG{\ProjG{i}{\TransExp{e}}}
\end{align*}
The recursion principle for the list recursor is more involved, due to the translation of lists as quadratic-credit lists.
Consider the recursor
$\FoldrIG{e_0}{a}{b}{e_1}{e} : \B[0]$.
Although we could naively translate to a recursor for quadratic-credit lists at the constant family $\Inf{\B}(r) \isdef \Inf{\B[0]}$,
there are two issues that arise in common examples.
Recall the typing judgment for quadratic-credit lists (\cref{sec:giralf:syntax:plistii}).
\begin{enumerate}
  \item
        In the base case $e_0$, no credits are made available, even though cost could be incurred in the original Calf program.
        To mitigate this issue, it is common to thread credits through using the debit operator, similar to the credit-passing aspect of \cref{lem:paramorphism}.

  \item
        In the inductive case $e_1$, some unknown quantity $r$ of credits is made available.
        However, if any of that is to be spent, it must be known that $r$ is large enough.
        For this reason, it is important to build an \emph{invariant} into the output family that restricts the possible values of $r$, representable using a power type $\PowerC{\VQ}{\CMP*{\cdots}}$.
\end{enumerate}
We solve both of these issues by choosing a family $\Inf{\B}(r) \isdef \PowerC{\VQ}{\Debit{\Q}{\Inf{\B[0]}}}$ (where both the unknown invariant $\VQ$ and the unknown debit annotation $\Q$ may depend on $r$) and define
\[ \TransExp{\FoldrIG{e_0}{x}{a}{e_1}{e}} \isdef \PayG{\FoldrIIG{r}{\Inf{\B}(r)}{\TransExp{e_0}'}{a}{b}{\TransExp{e_1}'}{\TransExp{e}}~\Q} \]
where
$\TransExp{e_0}' \isdef \plam{(h : \VQ)}{\GetG{\Q}{\TransExp{e_0}}}$ and
$\TransExp{e_1}' \isdef \plam{(h : \VQ)}{\GetG{\Q}{[\PayG{b~\Q}/b] \TransExp{e_1}}}$ thread the debits and invariant recursively through the code.
To correctly close the loop, the complementary elimination forms are then be applied to the recursive result $b$ and the entire $\KW{foldr}$ expression.

\begin{example}
  As a basic example, consider $\DEFN{snoc}$, a worst-case simplification of
  $\DEFN{insert}$ (\cref{ex:insert}) that appends an element $x : \NatV$ to the end
  of a list: %
  \par\iblock{
    \mrow{\DEFN{snoc} : \LolliV{\ListC[\F{\NatV}]}{\ListC[\F{\NatV}]}}
    \mrow{\DEFN{snoc}~l \isdef \FoldrIG{\ConsG{\ret{x}}{\NilG}}{(\ret{y})}{b}{\charge{1}[{\ConsG{\ret{y}}{b}}]}{l}}
  }\noindent
  This program translates to the following skeletal Giralf program:
  \par\iblock{
    \mrow{\Inf{\DEFN{snoc}} : \LolliV{\PListII{\Q}{\Q}[\F{\NatV}]}{\PListII{\Q}{\Q}[\F{\NatV}]}}
    \mhang{\Inf{\DEFN{snoc}}~l \isdef \PayG{\Inf{\DEFN{snoc'}}~?}~\KW{where}}{
        \mrow{e_0 \isdef \plam{(h : \VQ)}{\GetG{\Q}{\ConsG{\ret{x}}{\NilG}}}}
        \mrow{e_1 \isdef \plam{(h : \VQ)}{\GetG{\Q}{\ChargeG{1}[{\ConsG{\ret{y}}{(\PayG{b~\Q})}}]}}}
        \mrow{\Inf{\DEFN{snoc'}} \isdef \FoldrIIG{r}{\PowerC{\VQ}{\Debit{\Q}{\PListII{\Q}{\Q}[\F{\NatV}]}}}{e_0}{(\ret{y})}{b}{e_1}{l}}
    }
  }\noindent
  This skeletal Giralf code faithfully reconstructs $\DEFN{snoc}$ with the assumption that all cost is prepaid.
\end{example}

\subsubsection{Constraint Solving}

The crux of cost inference is determining the unknown amounts of credits, which will induce the final cost bound.
We achieve this by adapting the LP-based approach of RaML~\citep{hoffmann-aehlig-hofmann>2012-raml,hoffmann-hofmann>2010-aplas}, which performs cost inference on programs with unbounded recursion, to the present setting with structural recursion.

Because RaML is implemented using an LP solver, it always outputs resource bounds with concrete numbers rather than symbolic expressions.
This causes the precise technique to be incompatible with the quadratic-credit list recursor, which uses an inductive family $\B(r)$ due to the fact that the linear coefficient $r$ changes inductively.
In particular, the requisite invariant must be a symbolic expression involving $r$.
To surmount this problem, we introduce the notion of \emph{cost-free} and
\emph{cost-aware} annotations~\citep{hoffmann-aehlig-hofmann>2012}.
For a skeletal Giralf program $e$,
\begin{enumerate}
  \item a \emph{cost-aware annotation} is a valid numerical assignment of credit values to unknowns; and
  \item a \emph{cost-free annotation} is a cost-aware annotation for $\widehat{e}$, constructed by erasing all cost annotations in $e$, describing only how credits move from input to output.
\end{enumerate}

\begin{example}\label{ex:snoc-annotations}
  Let us continue the running example of $\DEFN{snoc}$. To solve for its unknown
  credit values, we run the RaML algorithm on the recursor
  $\Inf{\DEFN{snoc'}}$:
  \begin{enumerate}
    \item One \emph{cost-aware} annotation demonstrates that $\Inf{\DEFN{snoc'}}$ spends at most $|l|$ credits:
          \[
            \Giralf{l : \PListII{1}{0}[\F{\NatV}]}{0}{\Inf{\DEFN{snoc'}}}{\PowerC{\VQ}{\Debit{0}{\PListII{0}{0}[\F{\NatV}]}}}.
          \]

    \item One \emph{cost-free} annotation describes the movement of linear credits:
          \[
            \Giralf{l : \PListII{1}{0}[\F{\NatV}]}{0}{\Inf{\widehat{\DEFN{snoc'}}}}{\PowerC{\VQ}{\Debit{1}{\PListII{1}{0}[\F{\NatV}]}}}.
          \]
          Given input list $l$ carrying $\ListVLength{l}$ credits, $\Inf{\DEFN{snoc'}}$
          outputs a list $l'$ carrying $\ListVLength{l'}$ credits, but at the additional
          cost of 1 more credit (indicated by the debit operator) because $\ListVLength{l'} = \ListVLength{l} + 1$.

    \item Another \emph{cost-free} annotation describes the movement of quadratic credits:
          \[
            \Giralf{l : \PListII{1}{1}[\F{\NatV}]}{0}{\Inf{\widehat{\DEFN{snoc'}}}}{\PowerC{\VQ}{\Debit{0}{\PListII{0}{1}[\F{\NatV}]}}}.
          \]
  \end{enumerate}
  RaML is capable of automating all three of the inferences shown above.
\end{example}

Numerical cost-aware and cost-free types can be combined into a symbolic invariant based on the following two lemmas~\citep{hoffmann-hofmann>2010-aplas}.

\begin{lemma}
  \emph{Cost-free annotations} are closed under addition and scalar
  multiplication. That is, if $\Giralf{\DeltaC[1]}{q_1}{\Inf{\widehat{e}}}{\A[1]}$
  and $\Giralf{\DeltaC[2]}{q_2}{\Inf{\widehat{e}}}{\A[2]}$ are credit assignments on the same skeletal Giralf types, then
  \begin{align*}
    \Giralf{\DeltaC[1] \oplus \DeltaC[2]}{q_1 + q_2}{\Inf{\widehat{e}}}{\A[1] \oplus \A[2]}
    &&
    \text{ and }
    &&
    \Giralf{k \odot \DeltaC[1]}{k \odot q_1}{\Inf{\widehat{e}}}{k \odot \A[1]}
  \end{align*}
  where $\oplus$ and $\odot$ are pointwise addition and scalar multiplication of the assigned credit values. %
\end{lemma}

\begin{lemma}
  \emph{Cost-aware annotations} are closed under addition with a \emph{cost-free
  annotation} of the same term. That is, if $\Giralf{\DeltaC[1]}{q_1}{\Inf{e}}{\A[1]}$
  and $\Giralf{\DeltaC[2]}{q_2}{\Inf{\widehat{e}}}{\A[2]}$, then
  $
    \Giralf{\DeltaC[1] \oplus \DeltaC[2]}{q_1 + q_2}{\Inf{e}}{\A[1] \oplus \A[2]}
  $.
\end{lemma}

\begin{example}
  With these lemmas, we combine the numerical annotations of
  \cref{ex:snoc-annotations} into a single type with symbolic credit variables.
  Denoting the three types as $\A[\text{ca}]$, $\A[\text{cf1}]$, and $\A[\text{cf2}]$, we combine them via $\A[\text{ca}] \oplus (p_1 \odot \A[\text{cf1}]) \oplus (p_2 \odot
    \A[\text{cf2}])$ to obtain the desired resource-polymorphic typing
  \[
    \Giralf{l : \PListII{1 + p_1 + p_2}{p_2}[\F{\NatV}]}{0}{\Inf{\DEFN{snoc'}}}{\PowerC{\VQ}{\Debit{p_1}{\PListII{p_1}{p_2}[\F{\NatV}]}}}.
  \]
  Then, with a change of basis $r \isdef 1 + p_1 + p_2$, we determine the full type family for $\Inf{snoc'}$:
  \[
    \B(r) \isdef \PowerC{(r \ge 1 + p_2)}{\Debit{r - p_2 - 1}{\PListII{r - p_2 - 1}{p_2}[\F{\NatV}]}}.
  \]
  The invariant predicate guarantees that the changing linear coefficient $r$ is always large enough by ensuring all subtractions of credits are non-negative.
\end{example}

In summary, using the LP solving technique of AARA and RaML, we
infer valid invariants and credit amounts.
Once the unknown credit amounts within the skeletal Giralf \emph{types} are
solved, we populate the unknown credit amounts within the skeletal Giralf
\emph{terms} via a simple algorithm inspired by bidirectional type checking;
\eg, when checking $\GetG{\Q}{e}$ against the type
$\Debit{p}{\A}$, we replace $\Q$ with $p$.
The only remaining missing data are proofs that the symbolic invariants are
recursively preserved; we generate such proofs using a simple heuristic-driven
arithmetic solver.

\subsubsection{Cost Inference for Calf}

Using the above inference procedure and the semantics of Giralf in Calf, we may infer cost bounds on Calf programs automatically.

\begin{example}
  Let $\DEFN{isort} : \U[\PowerC{\LV}{\F[\LV]}]$ be the insertion sort algorithm in Calf~\citep{niu-sterling-grodin-harper>2022}.
  This program may be written with the equivalent types $\DEFN{isort} : \LolliV{\F[\LV]}{\F[\LV ]}$ and $\DEFN{isort} : \LolliV{\ListC[\F{\NatV}]}{\ListC[\F{\NatV}]}$.
  The previously described inference algorithm derives the judgment
  \[ \Giralf{\PListII{p_1}{p_2 + 1}[\F{\NatV}]}{0}{\Inf{\DEFN{isort}}}{\PListII{p_1}{p_2}[\F{\NatV}]} \]
  demonstrating that the quadratic credits stored alongside the input list decrease by $1$, indicating a triangular cost upper bound.
  Let $p_1 = p_2 = 0$; as a Calf program (up to \cref{lem:credit-monoidal}),
  \[ \Inf{\DEFN{isort}} : \LolliV{\PListII{0}{1}[\F{\NatV}]}{\F[\LV]}, \]
  and by \cref{thm:calf-aara-types,lem:PotentialC}, with $\Phi(l) \isdef \binom{\ListVLength{l}}{2}$ we have that
  \[ \Inf{\DEFN{isort}} : \LolliV{\CopowerC{(l : \LV)}{\Credit{\Phi(l)}{\TopC}}}{\F[\LV]}. \]
  So, pre-paying as in \cref{rem:prepay}, we find by \cref{lem:save-spend} that
  \[ \clam{(\ret{l})}{\Inf{\DEFN{isort}}(l, \CreditSave{\Phi(l)})} : \LolliV{\F[\LV]}{\F[\LV]} \]
  is an upper bound for the original $\DEFN{isort}$ program in Calf in the abstract phase.
\end{example}

\section{Related Work}\label{sec:related}

The central ideas on which the present work is based are verified cost analysis in Calf~\citep{niu-sterling-grodin-harper>2022,grodin-niu-sterling-harper>2024,grodin-harper>2024}, fracture and gluing for synthetic abstraction~\citep{rijke-shulman-spitters>2020,grodin-li-harper>2026}, and AARA~\citep{hofmann-jost>2003,hoffmann-jost>2022,das-balzer-hoffmann-pfenning-santurkar>2021}.
However, there are many other approaches to formally verifying amortized cost.

Using potential functions from the physicist's view, \citet{vanbrugge>2024} and \citet{nipkow-brinkop>2019} mechanize various amortized bounds.
\citet{atkey>2014} develops a connection between abstraction and conservation of energy which parallels the perspective taken in this work.

\citet{cutler-licata-danner>2020} propose a non-dependent graded language similar to Giralf whose semantics is also based on the cost effect (a writer monad).
\citet{danielsson>2008} verifies amortized costs in dependent type theory using a graded monad representing debits, and \citet{atkey>2011}, \citet{mevel-jourdan-pottier>2019}, and \citet{pottier-gueneau-jourdan-mevel>2024} extend separation logics with a credit resource.

\paragraph[λ-amor]{$\lambda$-amor}
\citet{rajani-barthe-garg>2024} present $\lambda$-amor, an extension of the linear type theory of AARA with a layer of structural refinements, roughly corresponding to the layer of value types in Giralf.
Although Giralf and $\lambda$-amor have some key technical differences---$\lambda$-amor includes, for example, unbounded recursion and impredicative quantification---the overall structures and goals of the type theories are aligned.
Notably, $\lambda$-amor is the first work based on the banker's method of which we are aware that includes negative types, such as (in our notation) products $\ProdC{\A}{\B}$, powers $\PowerC{\X}{\A}$, homomorphisms $\LolliC{\A}{\B}$, and debits $\Debit{c}{\A}$.
Semantically, this is achieved by extending the potential functions of AARA (reviewed in \cref{sec:giralf:semantics}) to potential \emph{predicates} $\Phi_{\A} : \ExpV{\AARAValues{\A}}{\ExpV{\Cost}{\VVProp}}$ that determine if a given amount of potential is sufficient to construct a particular value.
Glimmers of the lax commutative squares emphasized in the present work are visible in these predicates.

\paragraph{Dependent AARA}
More recently, \citet{xu-wang>2026} extended AARA to support potential that is explicitly dependent on data, written (adapting our notation) $\GammaV \mathbin{;} \Phi_{\GammaV} \vdash e : \X \mathbin{;} \Phi_{\X}$, where $\Phi_{\GammaV} : \ExpV{\GammaV}{\Cost}$ and $\Phi_{\X} : \ExpV{\X}{\Cost}$ describe the input and output potentials.
We conjecture that their lightweight dependent type system can be understood in our full dependent type theory by interpreting their typing judgment as the type of homomorphisms $\LolliV{\PotentialC{\GammaV}{\Phi_{\GammaV}}}{\PotentialC{\X}{\Phi_{\X}}}$.

\section{Conclusion} \label{sec:conclusion}

In this work, we render amortized analysis synthetically in the Calf dependent type theory, treating the cost of an effectful abstraction function built into a type as potential.
This ensures that every program comes equipped with an amortized cost interface coherent up to a generalization of the physicist's conservation of potential that accounts for both cost and behavior.
Defining the banker's credits and debits synthetically, we develop Giralf, a substructural dependent type theory layered on top of Calf streamlining programming with credits.
Extending the automated cost inference techniques of AARA, we present an automated cost inference procedure for Calf targeting Giralf.

\paragraph{Future Work}
In this work, we focus on amortization for ephemeral data structures.
\citet{okasaki>1999} uses laziness to adapt amortization to the persistent setting, developed logically by \citet{pottier-gueneau-jourdan-mevel>2024} and operationally by \citet{lorenzen>2026}.
It is yet unknown if the technique developed in the present work can be adapted to the setting of persistent amortization.

In this work, we made use of a commutative cost model, which was essential for the semantics of types such as tensor products and credits.
However, in AARA, a non-commutative cost model is sometimes used to study high-water marks on cost~\citep{hoffmann-hofmann>2010-aplas}.
We leave the incorporation of such cost models to future work.

The methods of cost inference in AARA have been developed extensively over decades~\citep{hoffmann-jost>2022}.
In this work, we adapt the inference algorithm for polynomial credits on lists~\citep{hoffmann-hofmann>2010}, but we leave it as future work to integrate features such as multivariate bounds~\citep{hoffmann-aehlig-hofmann>2011,hoffmann-aehlig-hofmann>2012}, exponential bounds~\citep{kahn-hoffmann>2020}, and more general inductive types~\citep{grosen-kahn-hoffmann>2023}.

\section*{Data Availability Statement}

The core ideas presented in this work are implemented and mechanized.
Fusing and extending the mechanizations of \citet{niu-sterling-grodin-harper>2022}, \citet{grodin-niu-sterling-harper>2024}, and \citet{grodin-li-harper>2026}, we provide a mechanization in Cubical Agda~\citep{norell>2009,vezzosi-mortberg-abel>2019} (indicated by \AgdaFormalized{}) of the central constructions and theorems of \cref{sec:physicist,sec:modularity,sec:banker} and the semantics of Giralf described in \cref{sec:giralf}.
As an approximation, the inequality structure of \citet{grodin-niu-sterling-harper>2024} is replaced with equality.

We also provide an implementation of the inference algorithm given in \cref{sec:giralf:inference}, adapting the OCaml code of RaML~\citep{hoffmann-aehlig-hofmann>2012-raml}; while this code is not verified, it is \emph{certifying}, emitting Giralf artifacts in Agda that include a certificate of amortized correctness by construction.

\begin{acks}
  The authors thank Reid Barton for his advice on the proof of \cref{thm:cfracglue}
  and Jonathan Sterling for fruitful adjacent collaboration.

  This material is based upon work supported by
  \grantsponsor{JSC}{Jane Street Group,
    LLC}{https://www.janestreet.com/};
  the \grantsponsor{AFOSR}{United States Air Force Office of
    Scientific Research}{https://www.afrl.af.mil/AFOSR/} under grant
  numbers \grantnum{AFOSR}{FA9550-21-0009} and
  \grantnum{AFOSR}{FA9550-23-1-0434} (Tristan Nguyen, program
  manager);
  and the \grantsponsor{NSF}{National Science
    Foundation}{https://nsf.gov} under award numbers
  \grantnum{NSF}{2615896}, \grantnum{NSF}{2311983}, and
  \grantnum{NSF}{2525102}.
  Any opinions, findings and conclusions or recommendations expressed
  in this material are those of the authors and do not necessarily
  reflect the views of the AFOSR and the NSF.
\end{acks}

\bibliographystyle{ACM-Reference-Format}
\bibliography{main}


\begin{thebibliography}{59}


\ifx \showCODEN    \undefined \def \showCODEN     #1{\unskip}     \fi
\ifx \showDOI      \undefined \def \showDOI       #1{#1}\fi
\ifx \showISBNx    \undefined \def \showISBNx     #1{\unskip}     \fi
\ifx \showISBNxiii \undefined \def \showISBNxiii  #1{\unskip}     \fi
\ifx \showISSN     \undefined \def \showISSN      #1{\unskip}     \fi
\ifx \showLCCN     \undefined \def \showLCCN      #1{\unskip}     \fi
\ifx \shownote     \undefined \def \shownote      #1{#1}          \fi
\ifx \showarticletitle \undefined \def \showarticletitle #1{#1}   \fi
\ifx \showURL      \undefined \def \showURL       {\relax}        \fi
\providecommand\bibfield[2]{#2}
\providecommand\bibinfo[2]{#2}
\providecommand\natexlab[1]{#1}
\providecommand\showeprint[2][]{arXiv:#2}

\bibitem[Ahman et~al\mbox{.}(2016)]%
        {ahman-ghani-plotkin>2016}
\bibfield{author}{\bibinfo{person}{Danel Ahman}, \bibinfo{person}{Neil Ghani}, {and} \bibinfo{person}{Gordon~D. Plotkin}.} \bibinfo{year}{2016}\natexlab{}.
\newblock \showarticletitle{Dependent {{Types}} and {{Fibred Computational Effects}}}. In \bibinfo{booktitle}{\emph{Foundations of {{Software Science}} and {{Computation Structures}}}} \emph{(\bibinfo{series}{Lecture {{Notes}} in {{Computer Science}}})}, \bibfield{editor}{\bibinfo{person}{Bart Jacobs} {and} \bibinfo{person}{Christof L{\"o}ding}} (Eds.). \bibinfo{publisher}{Springer}, \bibinfo{address}{Berlin, Heidelberg}, \bibinfo{pages}{36--54}.
\newblock
\showISBNx{978-3-662-49630-5}
\urldef\tempurl%
\url{https://doi.org/10.1007/978-3-662-49630-5_3}
\showDOI{\tempurl}


\bibitem[Atkey(2011)]%
        {atkey>2011}
\bibfield{author}{\bibinfo{person}{Robert Atkey}.} \bibinfo{year}{2011}\natexlab{}.
\newblock \showarticletitle{Amortised {{Resource Analysis}} with {{Separation Logic}}}.
\newblock \bibinfo{journal}{\emph{Logical Methods in Computer Science}}  \bibinfo{volume}{Volume 7, Issue 2} (\bibinfo{date}{June} \bibinfo{year}{2011}).
\newblock
\showISSN{1860-5974}
\urldef\tempurl%
\url{https://doi.org/10.2168/LMCS-7(2:17)2011}
\showDOI{\tempurl}


\bibitem[Atkey(2014)]%
        {atkey>2014}
\bibfield{author}{\bibinfo{person}{Robert Atkey}.} \bibinfo{year}{2014}\natexlab{}.
\newblock \showarticletitle{From Parametricity to Conservation Laws, via {{Noether}}'s Theorem}. In \bibinfo{booktitle}{\emph{Proceedings of the 41st {{ACM SIGPLAN-SIGACT Symposium}} on {{Principles}} of {{Programming Languages}}}} \emph{(\bibinfo{series}{{{POPL}} '14})}. \bibinfo{publisher}{Association for Computing Machinery}, \bibinfo{address}{New York, NY, USA}, \bibinfo{pages}{491--502}.
\newblock
\showISBNx{978-1-4503-2544-8}
\urldef\tempurl%
\url{https://doi.org/10.1145/2535838.2535867}
\showDOI{\tempurl}


\bibitem[Benton(1995)]%
        {benton>1995}
\bibfield{author}{\bibinfo{person}{P.~N. Benton}.} \bibinfo{year}{1995}\natexlab{}.
\newblock \showarticletitle{A Mixed Linear and Non-Linear Logic: {{Proofs}}, Terms and Models}. In \bibinfo{booktitle}{\emph{Computer {{Science Logic}}}} \emph{(\bibinfo{series}{Lecture {{Notes}} in {{Computer Science}}})}, \bibfield{editor}{\bibinfo{person}{Leszek Pacholski} {and} \bibinfo{person}{Jerzy Tiuryn}} (Eds.). \bibinfo{publisher}{Springer}, \bibinfo{address}{Berlin, Heidelberg}, \bibinfo{pages}{121--135}.
\newblock
\showISBNx{978-3-540-49404-1}
\urldef\tempurl%
\url{https://doi.org/10.1007/BFb0022251}
\showDOI{\tempurl}


\bibitem[Burton(1982)]%
        {burton>1982}
\bibfield{author}{\bibinfo{person}{F.~Warren Burton}.} \bibinfo{year}{1982}\natexlab{}.
\newblock \showarticletitle{An Efficient Functional Implementation of {{FIFO}} Queues}.
\newblock \bibinfo{journal}{\emph{Inform. Process. Lett.}} \bibinfo{volume}{14}, \bibinfo{number}{5} (\bibinfo{date}{July} \bibinfo{year}{1982}), \bibinfo{pages}{205--206}.
\newblock
\showISSN{0020-0190}
\urldef\tempurl%
\url{https://doi.org/10.1016/0020-0190(82)90015-1}
\showDOI{\tempurl}


\bibitem[Cutler et~al\mbox{.}(2020)]%
        {cutler-licata-danner>2020}
\bibfield{author}{\bibinfo{person}{Joseph~W. Cutler}, \bibinfo{person}{Daniel~R. Licata}, {and} \bibinfo{person}{Norman Danner}.} \bibinfo{year}{2020}\natexlab{}.
\newblock \showarticletitle{Denotational Recurrence Extraction for Amortized Analysis}.
\newblock \bibinfo{journal}{\emph{Proceedings of the ACM on Programming Languages}} \bibinfo{volume}{4}, \bibinfo{number}{ICFP} (\bibinfo{date}{Aug.} \bibinfo{year}{2020}), \bibinfo{pages}{97:1--97:29}.
\newblock
\urldef\tempurl%
\url{https://doi.org/10.1145/3408979}
\showDOI{\tempurl}


\bibitem[Danielsson(2008)]%
        {danielsson>2008}
\bibfield{author}{\bibinfo{person}{Nils~Anders Danielsson}.} \bibinfo{year}{2008}\natexlab{}.
\newblock \showarticletitle{Lightweight Semiformal Time Complexity Analysis for Purely Functional Data Structures}.
\newblock \bibinfo{journal}{\emph{ACM SIGPLAN Notices}} \bibinfo{volume}{43}, \bibinfo{number}{1} (\bibinfo{date}{Jan.} \bibinfo{year}{2008}), \bibinfo{pages}{133--144}.
\newblock
\showISSN{0362-1340}
\urldef\tempurl%
\url{https://doi.org/10.1145/1328897.1328457}
\showDOI{\tempurl}


\bibitem[Das et~al\mbox{.}(2021)]%
        {das-balzer-hoffmann-pfenning-santurkar>2021}
\bibfield{author}{\bibinfo{person}{Ankush Das}, \bibinfo{person}{Stephanie Balzer}, \bibinfo{person}{Jan Hoffmann}, \bibinfo{person}{Frank Pfenning}, {and} \bibinfo{person}{Ishani Santurkar}.} \bibinfo{year}{2021}\natexlab{}.
\newblock \showarticletitle{Resource-{{Aware Session Types}} for {{Digital Contracts}}}. In \bibinfo{booktitle}{\emph{2021 {{IEEE}} 34th {{Computer Security Foundations Symposium}} ({{CSF}})}}. \bibinfo{pages}{1--16}.
\newblock
\showISSN{2374-8303}
\urldef\tempurl%
\url{https://doi.org/10.1109/CSF51468.2021.00004}
\showDOI{\tempurl}


\bibitem[Egger et~al\mbox{.}(2009)]%
        {egger-mogelberg-simpson>2009}
\bibfield{author}{\bibinfo{person}{Jeff Egger}, \bibinfo{person}{Rasmus~Ejlers M{\o}gelberg}, {and} \bibinfo{person}{Alex Simpson}.} \bibinfo{year}{2009}\natexlab{}.
\newblock \showarticletitle{Enriching an {{Effect Calculus}} with {{Linear Types}}}. In \bibinfo{booktitle}{\emph{Computer {{Science Logic}}}} \emph{(\bibinfo{series}{Lecture {{Notes}} in {{Computer Science}}})}, \bibfield{editor}{\bibinfo{person}{Erich Gr{\"a}del} {and} \bibinfo{person}{Reinhard Kahle}} (Eds.). \bibinfo{publisher}{Springer}, \bibinfo{address}{Berlin, Heidelberg}, \bibinfo{pages}{240--254}.
\newblock
\showISBNx{978-3-642-04027-6}
\urldef\tempurl%
\url{https://doi.org/10.1007/978-3-642-04027-6_19}
\showDOI{\tempurl}


\bibitem[Egger et~al\mbox{.}(2014)]%
        {egger-mogelberg-simpson>2014}
\bibfield{author}{\bibinfo{person}{Jeff Egger}, \bibinfo{person}{Rasmus~Ejlers M{\o}gelberg}, {and} \bibinfo{person}{Alex Simpson}.} \bibinfo{year}{2014}\natexlab{}.
\newblock \showarticletitle{The Enriched Effect Calculus: Syntax and Semantics}.
\newblock \bibinfo{journal}{\emph{Journal of Logic and Computation}} \bibinfo{volume}{24}, \bibinfo{number}{3} (\bibinfo{date}{June} \bibinfo{year}{2014}), \bibinfo{pages}{615--654}.
\newblock
\showISSN{1465-363X}
\urldef\tempurl%
\url{https://doi.org/10.1093/logcom/exs025}
\showDOI{\tempurl}


\bibitem[Grandis and Par{\'e}(2004)]%
        {grandis-pare>2004}
\bibfield{author}{\bibinfo{person}{Marco Grandis} {and} \bibinfo{person}{Robert Par{\'e}}.} \bibinfo{year}{2004}\natexlab{}.
\newblock \showarticletitle{Adjoint for Double Categories}.
\newblock \bibinfo{journal}{\emph{Cahiers de Topologie et G\'eom\'etrie Diff\'erentielle Cat\'egoriques}} \bibinfo{volume}{45}, \bibinfo{number}{3} (\bibinfo{year}{2004}), \bibinfo{pages}{193--240}.
\newblock
\showISSN{2681-2363}
\urldef\tempurl%
\url{https://www.numdam.org/item/?id=CTGDC_2004__45_3_193_0}
\showURL{%
\tempurl}


\bibitem[Gries(1989)]%
        {gries>1989}
\bibfield{author}{\bibinfo{person}{David Gries}.} \bibinfo{year}{1989}\natexlab{}.
\newblock \bibinfo{booktitle}{\emph{The {{Science}} of {{Programming}}}}.
\newblock \bibinfo{publisher}{Springer New York}.
\newblock
\showISBNx{978-0-387-96480-5}


\bibitem[Grodin and Harper(2024)]%
        {grodin-harper>2024}
\bibfield{author}{\bibinfo{person}{Harrison Grodin} {and} \bibinfo{person}{Robert Harper}.} \bibinfo{year}{2024}\natexlab{}.
\newblock \showarticletitle{Amortized {{Analysis}} via {{Coalgebra}}}.
\newblock \bibinfo{journal}{\emph{Electronic Notes in Theoretical Informatics and Computer Science}}  \bibinfo{volume}{Volume 4 - Proceedings of MFPS XL} (\bibinfo{date}{Dec.} \bibinfo{year}{2024}).
\newblock
\showISSN{2969-2431}
\urldef\tempurl%
\url{https://doi.org/10.46298/entics.14797}
\showDOI{\tempurl}


\bibitem[Grodin et~al\mbox{.}(2026)]%
        {grodin-li-harper>2026}
\bibfield{author}{\bibinfo{person}{Harrison Grodin}, \bibinfo{person}{Runming Li}, {and} \bibinfo{person}{Robert Harper}.} \bibinfo{year}{2026}\natexlab{}.
\newblock \showarticletitle{Abstraction {{Functions}} as {{Types}}: {{Modular Verification}} of {{Cost}} and {{Behavior}} in {{Dependent Type Theory}}}.
\newblock \bibinfo{journal}{\emph{Proceedings of the ACM on Programming Languages}} \bibinfo{volume}{10}, \bibinfo{number}{POPL} (\bibinfo{date}{Jan.} \bibinfo{year}{2026}), \bibinfo{pages}{31:895--31:922}.
\newblock
\urldef\tempurl%
\url{https://doi.org/10.1145/3776673}
\showDOI{\tempurl}


\bibitem[Grodin et~al\mbox{.}(2024)]%
        {grodin-niu-sterling-harper>2024}
\bibfield{author}{\bibinfo{person}{Harrison Grodin}, \bibinfo{person}{Yue Niu}, \bibinfo{person}{Jonathan Sterling}, {and} \bibinfo{person}{Robert Harper}.} \bibinfo{year}{2024}\natexlab{}.
\newblock \showarticletitle{Decalf: {{A Directed}}, {{Effectful Cost-Aware Logical Framework}}}.
\newblock \bibinfo{journal}{\emph{Proceedings of the ACM on Programming Languages}} \bibinfo{volume}{8}, \bibinfo{number}{POPL} (\bibinfo{date}{Jan.} \bibinfo{year}{2024}), \bibinfo{pages}{10:273--10:301}.
\newblock
\urldef\tempurl%
\url{https://doi.org/10.1145/3632852}
\showDOI{\tempurl}


\bibitem[Grosen et~al\mbox{.}(2023)]%
        {grosen-kahn-hoffmann>2023}
\bibfield{author}{\bibinfo{person}{Jessie Grosen}, \bibinfo{person}{David~M. Kahn}, {and} \bibinfo{person}{Jan Hoffmann}.} \bibinfo{year}{2023}\natexlab{}.
\newblock \showarticletitle{Automatic {{Amortized Resource Analysis}} with {{Regular Recursive Types}}}. In \bibinfo{booktitle}{\emph{2023 38th {{Annual ACM}}/{{IEEE Symposium}} on {{Logic}} in {{Computer Science}} ({{LICS}})}}. \bibinfo{pages}{1--14}.
\newblock
\urldef\tempurl%
\url{https://doi.org/10.1109/LICS56636.2023.10175720}
\showDOI{\tempurl}


\bibitem[Guibas and Sedgewick(1978)]%
        {guibas-sedgewick>1978}
\bibfield{author}{\bibinfo{person}{Leo~J. Guibas} {and} \bibinfo{person}{Robert Sedgewick}.} \bibinfo{year}{1978}\natexlab{}.
\newblock \showarticletitle{A Dichromatic Framework for Balanced Trees}. In \bibinfo{booktitle}{\emph{19th {{Annual Symposium}} on {{Foundations}} of {{Computer Science}} (Sfcs 1978)}}. \bibinfo{pages}{8--21}.
\newblock
\showISSN{0272-5428}
\urldef\tempurl%
\url{https://doi.org/10.1109/SFCS.1978.3}
\showDOI{\tempurl}


\bibitem[Hoare(1972)]%
        {hoare>1972}
\bibfield{author}{\bibinfo{person}{C.~A.~R. Hoare}.} \bibinfo{year}{1972}\natexlab{}.
\newblock \showarticletitle{Proof of {{Correctness}} of {{Data Representations}}}.
\newblock \bibinfo{journal}{\emph{Acta Informatica}} \bibinfo{volume}{1}, \bibinfo{number}{4} (\bibinfo{date}{Dec.} \bibinfo{year}{1972}), \bibinfo{pages}{271--281}.
\newblock
\showISSN{1432-0525}
\urldef\tempurl%
\url{https://doi.org/10.1007/BF00289507}
\showDOI{\tempurl}


\bibitem[Hoffmann et~al\mbox{.}(2011)]%
        {hoffmann-aehlig-hofmann>2011}
\bibfield{author}{\bibinfo{person}{Jan Hoffmann}, \bibinfo{person}{Klaus Aehlig}, {and} \bibinfo{person}{Martin Hofmann}.} \bibinfo{year}{2011}\natexlab{}.
\newblock \showarticletitle{Multivariate Amortized Resource Analysis}. In \bibinfo{booktitle}{\emph{Proceedings of the 38th Annual {{ACM SIGPLAN-SIGACT}} Symposium on {{Principles}} of Programming Languages}} \emph{(\bibinfo{series}{{{POPL}} '11})}. \bibinfo{publisher}{Association for Computing Machinery}, \bibinfo{address}{New York, NY, USA}, \bibinfo{pages}{357--370}.
\newblock
\showISBNx{978-1-4503-0490-0}
\urldef\tempurl%
\url{https://doi.org/10.1145/1926385.1926427}
\showDOI{\tempurl}


\bibitem[Hoffmann et~al\mbox{.}(2012a)]%
        {hoffmann-aehlig-hofmann>2012}
\bibfield{author}{\bibinfo{person}{Jan Hoffmann}, \bibinfo{person}{Klaus Aehlig}, {and} \bibinfo{person}{Martin Hofmann}.} \bibinfo{year}{2012}\natexlab{a}.
\newblock \showarticletitle{Multivariate Amortized Resource Analysis}.
\newblock \bibinfo{journal}{\emph{ACM Transactions on Programming Languages and Systems}} \bibinfo{volume}{34}, \bibinfo{number}{3} (\bibinfo{date}{Nov.} \bibinfo{year}{2012}), \bibinfo{pages}{14:1--14:62}.
\newblock
\showISSN{0164-0925}
\urldef\tempurl%
\url{https://doi.org/10.1145/2362389.2362393}
\showDOI{\tempurl}


\bibitem[Hoffmann et~al\mbox{.}(2012b)]%
        {hoffmann-aehlig-hofmann>2012-raml}
\bibfield{author}{\bibinfo{person}{Jan Hoffmann}, \bibinfo{person}{Klaus Aehlig}, {and} \bibinfo{person}{Martin Hofmann}.} \bibinfo{year}{2012}\natexlab{b}.
\newblock \showarticletitle{Resource {{Aware ML}}}. In \bibinfo{booktitle}{\emph{Computer {{Aided Verification}}}} \emph{(\bibinfo{series}{Lecture {{Notes}} in {{Computer Science}}})}, \bibfield{editor}{\bibinfo{person}{P.~Madhusudan} {and} \bibinfo{person}{Sanjit~A. Seshia}} (Eds.). \bibinfo{publisher}{Springer}, \bibinfo{address}{Berlin, Heidelberg}, \bibinfo{pages}{781--786}.
\newblock
\showISBNx{978-3-642-31424-7}
\urldef\tempurl%
\url{https://doi.org/10.1007/978-3-642-31424-7_64}
\showDOI{\tempurl}


\bibitem[Hoffmann and Hofmann(2010a)]%
        {hoffmann-hofmann>2010-aplas}
\bibfield{author}{\bibinfo{person}{Jan Hoffmann} {and} \bibinfo{person}{Martin Hofmann}.} \bibinfo{year}{2010}\natexlab{a}.
\newblock \showarticletitle{Amortized {{Resource Analysis}} with {{Polymorphic Recursion}} and {{Partial Big-Step Operational Semantics}}}. In \bibinfo{booktitle}{\emph{Programming {{Languages}} and {{Systems}}}} \emph{(\bibinfo{series}{Lecture {{Notes}} in {{Computer Science}}})}, \bibfield{editor}{\bibinfo{person}{Kazunori Ueda}} (Ed.). \bibinfo{publisher}{Springer}, \bibinfo{address}{Berlin, Heidelberg}, \bibinfo{pages}{172--187}.
\newblock
\showISBNx{978-3-642-17164-2}
\urldef\tempurl%
\url{https://doi.org/10.1007/978-3-642-17164-2_13}
\showDOI{\tempurl}


\bibitem[Hoffmann and Hofmann(2010b)]%
        {hoffmann-hofmann>2010}
\bibfield{author}{\bibinfo{person}{Jan Hoffmann} {and} \bibinfo{person}{Martin Hofmann}.} \bibinfo{year}{2010}\natexlab{b}.
\newblock \showarticletitle{Amortized {{Resource Analysis}} with {{Polynomial Potential}}}. In \bibinfo{booktitle}{\emph{Programming {{Languages}} and {{Systems}}}}, \bibfield{editor}{\bibinfo{person}{Andrew~D. Gordon}} (Ed.). \bibinfo{publisher}{Springer}, \bibinfo{address}{Berlin, Heidelberg}, \bibinfo{pages}{287--306}.
\newblock
\showISBNx{978-3-642-11957-6}
\urldef\tempurl%
\url{https://doi.org/10.1007/978-3-642-11957-6_16}
\showDOI{\tempurl}


\bibitem[Hoffmann and Jost(2022)]%
        {hoffmann-jost>2022}
\bibfield{author}{\bibinfo{person}{Jan Hoffmann} {and} \bibinfo{person}{Steffen Jost}.} \bibinfo{year}{2022}\natexlab{}.
\newblock \showarticletitle{Two Decades of Automatic Amortized Resource Analysis}.
\newblock \bibinfo{journal}{\emph{Mathematical Structures in Computer Science}} \bibinfo{volume}{32}, \bibinfo{number}{6} (\bibinfo{date}{June} \bibinfo{year}{2022}), \bibinfo{pages}{729--759}.
\newblock
\showISSN{0960-1295, 1469-8072}
\urldef\tempurl%
\url{https://doi.org/10.1017/S0960129521000487}
\showDOI{\tempurl}


\bibitem[Hofmann and Jost(2003)]%
        {hofmann-jost>2003}
\bibfield{author}{\bibinfo{person}{Martin Hofmann} {and} \bibinfo{person}{Steffen Jost}.} \bibinfo{year}{2003}\natexlab{}.
\newblock \showarticletitle{Static Prediction of Heap Space Usage for First-Order Functional Programs}.
\newblock \bibinfo{journal}{\emph{ACM SIGPLAN Notices}} \bibinfo{volume}{38}, \bibinfo{number}{1} (\bibinfo{date}{Jan.} \bibinfo{year}{2003}), \bibinfo{pages}{185--197}.
\newblock
\showISSN{0362-1340}
\urldef\tempurl%
\url{https://doi.org/10.1145/640128.604148}
\showDOI{\tempurl}


\bibitem[Hofmann et~al\mbox{.}(2022)]%
        {hofmann-leutgeb-obwaller-moser-zuleger>2022}
\bibfield{author}{\bibinfo{person}{Martin Hofmann}, \bibinfo{person}{Lorenz Leutgeb}, \bibinfo{person}{David Obwaller}, \bibinfo{person}{Georg Moser}, {and} \bibinfo{person}{Florian Zuleger}.} \bibinfo{year}{2022}\natexlab{}.
\newblock \showarticletitle{Type-Based Analysis of Logarithmic Amortised Complexity}.
\newblock \bibinfo{journal}{\emph{Mathematical Structures in Computer Science}} \bibinfo{volume}{32}, \bibinfo{number}{6} (\bibinfo{date}{June} \bibinfo{year}{2022}), \bibinfo{pages}{794--826}.
\newblock
\showISSN{0960-1295, 1469-8072}
\urldef\tempurl%
\url{https://doi.org/10.1017/S0960129521000232}
\showDOI{\tempurl}


\bibitem[Hood and Melville(1981)]%
        {hood-melville>1981}
\bibfield{author}{\bibinfo{person}{Robert Hood} {and} \bibinfo{person}{Robert Melville}.} \bibinfo{year}{1981}\natexlab{}.
\newblock \showarticletitle{Real-Time Queue Operations in Pure {{LISP}}}.
\newblock \bibinfo{journal}{\emph{Inform. Process. Lett.}} \bibinfo{volume}{13}, \bibinfo{number}{2} (\bibinfo{date}{Nov.} \bibinfo{year}{1981}), \bibinfo{pages}{50--54}.
\newblock
\showISSN{0020-0190}
\urldef\tempurl%
\url{https://doi.org/10.1016/0020-0190(81)90030-2}
\showDOI{\tempurl}


\bibitem[Kahn and Hoffmann(2020)]%
        {kahn-hoffmann>2020}
\bibfield{author}{\bibinfo{person}{David~M. Kahn} {and} \bibinfo{person}{Jan Hoffmann}.} \bibinfo{year}{2020}\natexlab{}.
\newblock \showarticletitle{Exponential {{Automatic Amortized Resource Analysis}}}. In \bibinfo{booktitle}{\emph{Foundations of {{Software Science}} and {{Computation Structures}}}} \emph{(\bibinfo{series}{Lecture {{Notes}} in {{Computer Science}}})}, \bibfield{editor}{\bibinfo{person}{Jean {Goubault-Larrecq}} {and} \bibinfo{person}{Barbara K{\"o}nig}} (Eds.). \bibinfo{publisher}{Springer International Publishing}, \bibinfo{address}{Cham}, \bibinfo{pages}{359--380}.
\newblock
\showISBNx{978-3-030-45231-5}
\urldef\tempurl%
\url{https://doi.org/10.1007/978-3-030-45231-5_19}
\showDOI{\tempurl}


\bibitem[Kebuladze(2025)]%
        {kebuladze>2025}
\bibfield{author}{\bibinfo{person}{Lukas Kebuladze}.} \bibinfo{year}{2025}\natexlab{}.
\newblock \bibinfo{booktitle}{\emph{Formally {{Verified Amortized Cost Analysis}} of {{Splay Trees}} in {{Agda}}}}.
\newblock \bibinfo{type}{{T}echnical {R}eport}. \bibinfo{institution}{Carnegie Mellon University}.
\newblock
\urldef\tempurl%
\url{https://www.cs.cmu.edu/~rwh/code/kebuladze_splay_tree.tar.gz}
\showURL{%
\tempurl}


\bibitem[Krishnaswami et~al\mbox{.}(2015)]%
        {krishnaswami-pradic-benton>2015}
\bibfield{author}{\bibinfo{person}{Neelakantan~R. Krishnaswami}, \bibinfo{person}{Pierre Pradic}, {and} \bibinfo{person}{Nick Benton}.} \bibinfo{year}{2015}\natexlab{}.
\newblock \showarticletitle{Integrating {{Linear}} and {{Dependent Types}}}.
\newblock \bibinfo{journal}{\emph{ACM SIGPLAN Notices}} \bibinfo{volume}{50}, \bibinfo{number}{1} (\bibinfo{year}{2015}), \bibinfo{pages}{17--30}.
\newblock
\showISBNx{9781450333009}
\showISSN{0362-1340}
\urldef\tempurl%
\url{https://doi.org/10.1145/2775051.2676969}
\showDOI{\tempurl}


\bibitem[Levy(2003)]%
        {levy>2003}
\bibfield{author}{\bibinfo{person}{Paul~Blain Levy}.} \bibinfo{year}{2003}\natexlab{}.
\newblock \bibinfo{booktitle}{\emph{Call-{{By-Push-Value}}: {{A Functional}}/{{Imperative Synthesis}}}}.
\newblock \bibinfo{publisher}{Springer Netherlands}, \bibinfo{address}{Dordrecht}.
\newblock
\showISBNx{978-94-010-3752-5 978-94-007-0954-6}
\urldef\tempurl%
\url{https://doi.org/10.1007/978-94-007-0954-6}
\showDOI{\tempurl}


\bibitem[Li and Harper(2025)]%
        {li-harper>2025}
\bibfield{author}{\bibinfo{person}{Runming Li} {and} \bibinfo{person}{Robert Harper}.} \bibinfo{year}{2025}\natexlab{}.
\newblock \bibinfo{title}{Canonicity for {{Cost-Aware Logical Framework}} via {{Synthetic Tait Computability}}}.
\newblock
\newblock
\urldef\tempurl%
\url{https://doi.org/10.48550/arXiv.2504.12464}
\showDOI{\tempurl}
\showeprint[arxiv]{2504.12464}~[cs]


\bibitem[Lorenzen(2026)]%
        {lorenzen>2026}
\bibfield{author}{\bibinfo{person}{Anton Lorenzen}.} \bibinfo{year}{2026}\natexlab{}.
\newblock \bibinfo{title}{Persistent {{Amortised Analysis}}, {{Operationally}}}.
\newblock
\newblock
\urldef\tempurl%
\url{https://doi.org/arXiv:2605.09411}
\showDOI{\tempurl}


\bibitem[Meertens(1992)]%
        {meertens>1992}
\bibfield{author}{\bibinfo{person}{Lambert Meertens}.} \bibinfo{year}{1992}\natexlab{}.
\newblock \showarticletitle{Paramorphisms}.
\newblock \bibinfo{journal}{\emph{Formal Aspects of Computing}} \bibinfo{volume}{4}, \bibinfo{number}{5} (\bibinfo{date}{Sept.} \bibinfo{year}{1992}), \bibinfo{pages}{413--424}.
\newblock
\showISSN{1433-299X}
\urldef\tempurl%
\url{https://doi.org/10.1007/BF01211391}
\showDOI{\tempurl}


\bibitem[M{\'e}vel et~al\mbox{.}(2019)]%
        {mevel-jourdan-pottier>2019}
\bibfield{author}{\bibinfo{person}{Glen M{\'e}vel}, \bibinfo{person}{Jacques-Henri Jourdan}, {and} \bibinfo{person}{Fran{\c c}ois Pottier}.} \bibinfo{year}{2019}\natexlab{}.
\newblock \showarticletitle{Time {{Credits}} and {{Time Receipts}} in {{Iris}}}. In \bibinfo{booktitle}{\emph{Programming {{Languages}} and {{Systems}}}}, \bibfield{editor}{\bibinfo{person}{Lu{\'i}s Caires}} (Ed.). \bibinfo{publisher}{Springer International Publishing}, \bibinfo{address}{Cham}, \bibinfo{pages}{3--29}.
\newblock
\showISBNx{978-3-030-17184-1}
\urldef\tempurl%
\url{https://doi.org/10.1007/978-3-030-17184-1_1}
\showDOI{\tempurl}


\bibitem[Nipkow and Brinkop(2019)]%
        {nipkow-brinkop>2019}
\bibfield{author}{\bibinfo{person}{Tobias Nipkow} {and} \bibinfo{person}{Hauke Brinkop}.} \bibinfo{year}{2019}\natexlab{}.
\newblock \showarticletitle{Amortized {{Complexity Verified}}}.
\newblock \bibinfo{journal}{\emph{Journal of Automated Reasoning}} \bibinfo{volume}{62}, \bibinfo{number}{3} (\bibinfo{date}{March} \bibinfo{year}{2019}), \bibinfo{pages}{367--391}.
\newblock
\showISSN{1573-0670}
\urldef\tempurl%
\url{https://doi.org/10.1007/s10817-018-9459-3}
\showDOI{\tempurl}


\bibitem[Niu and Harper(2022)]%
        {niu-harper>2022}
\bibfield{author}{\bibinfo{person}{Yue Niu} {and} \bibinfo{person}{Robert Harper}.} \bibinfo{year}{2022}\natexlab{}.
\newblock \bibinfo{title}{A Metalanguage for Cost-Aware Denotational Semantics}.
\newblock
\newblock
\urldef\tempurl%
\url{https://doi.org/10.48550/arXiv.2209.12669}
\showDOI{\tempurl}
\showeprint[arxiv]{2209.12669}~[cs]


\bibitem[Niu et~al\mbox{.}(2022)]%
        {niu-sterling-grodin-harper>2022}
\bibfield{author}{\bibinfo{person}{Yue Niu}, \bibinfo{person}{Jonathan Sterling}, \bibinfo{person}{Harrison Grodin}, {and} \bibinfo{person}{Robert Harper}.} \bibinfo{year}{2022}\natexlab{}.
\newblock \showarticletitle{A {{Cost-Aware Logical Framework}}}.
\newblock \bibinfo{journal}{\emph{Proceedings of the ACM on Programming Languages}} \bibinfo{volume}{6}, \bibinfo{number}{POPL} (\bibinfo{date}{Jan.} \bibinfo{year}{2022}), \bibinfo{pages}{9:1--9:31}.
\newblock
\urldef\tempurl%
\url{https://doi.org/10.1145/3498670}
\showDOI{\tempurl}


\bibitem[Norell(2009)]%
        {norell>2009}
\bibfield{author}{\bibinfo{person}{Ulf Norell}.} \bibinfo{year}{2009}\natexlab{}.
\newblock \showarticletitle{Dependently Typed Programming in {{Agda}}}. In \bibinfo{booktitle}{\emph{Proceedings of the 4th International Workshop on {{Types}} in Language Design and Implementation}} \emph{(\bibinfo{series}{{{TLDI}} '09})}. \bibinfo{publisher}{Association for Computing Machinery}, \bibinfo{address}{New York, NY, USA}, \bibinfo{pages}{1--2}.
\newblock
\showISBNx{978-1-60558-420-1}
\urldef\tempurl%
\url{https://doi.org/10.1145/1481861.1481862}
\showDOI{\tempurl}


\bibitem[Okasaki(1999)]%
        {okasaki>1999}
\bibfield{author}{\bibinfo{person}{Chris Okasaki}.} \bibinfo{year}{1999}\natexlab{}.
\newblock \bibinfo{booktitle}{\emph{Purely {{Functional Data Structures}}}}.
\newblock \bibinfo{publisher}{Cambridge University Press}.
\newblock
\showISBNx{978-0-521-66350-2}


\bibitem[P{\'e}drot and Tabareau(2019)]%
        {pedrot-tabareau>2019}
\bibfield{author}{\bibinfo{person}{Pierre-Marie P{\'e}drot} {and} \bibinfo{person}{Nicolas Tabareau}.} \bibinfo{year}{2019}\natexlab{}.
\newblock \showarticletitle{The {{Fire Triangle}}: {{How}} to {{Mix Substitution}}, {{Dependent Elimination}}, and {{Effects}}}.
\newblock \bibinfo{journal}{\emph{Proceedings of the ACM on Programming Languages}} \bibinfo{volume}{4}, \bibinfo{number}{POPL} (\bibinfo{date}{Dec.} \bibinfo{year}{2019}), \bibinfo{pages}{58:1--58:28}.
\newblock
\urldef\tempurl%
\url{https://doi.org/10.1145/3371126}
\showDOI{\tempurl}


\bibitem[Pham et~al\mbox{.}(2025)]%
        {pham-niu-glover-saad-hoffmann>2025}
\bibfield{author}{\bibinfo{person}{Long Pham}, \bibinfo{person}{Yue Niu}, \bibinfo{person}{Nathan Glover}, \bibinfo{person}{Feras Saad}, {and} \bibinfo{person}{Jan Hoffmann}.} \bibinfo{year}{2025}\natexlab{}.
\newblock \showarticletitle{Integrating {{Resource Analyses}} via {{Resource Decomposition}}}.
\newblock \bibinfo{journal}{\emph{Proceedings of the ACM on Programming Languages}} \bibinfo{volume}{9}, \bibinfo{number}{OOPSLA2} (\bibinfo{date}{Oct.} \bibinfo{year}{2025}), \bibinfo{pages}{409:3811--409:3840}.
\newblock
\urldef\tempurl%
\url{https://doi.org/10.1145/3763798}
\showDOI{\tempurl}


\bibitem[Pottier et~al\mbox{.}(2024)]%
        {pottier-gueneau-jourdan-mevel>2024}
\bibfield{author}{\bibinfo{person}{Fran{\c c}ois Pottier}, \bibinfo{person}{Arma{\"e}l Gu{\'e}neau}, \bibinfo{person}{Jacques-Henri Jourdan}, {and} \bibinfo{person}{Glen M{\'e}vel}.} \bibinfo{year}{2024}\natexlab{}.
\newblock \showarticletitle{Thunks and {{Debits}} in {{Separation Logic}} with {{Time Credits}}}.
\newblock \bibinfo{journal}{\emph{Proceedings of the ACM on Programming Languages}} \bibinfo{volume}{8}, \bibinfo{number}{POPL} (\bibinfo{date}{Jan.} \bibinfo{year}{2024}), \bibinfo{pages}{50:1482--50:1508}.
\newblock
\urldef\tempurl%
\url{https://doi.org/10.1145/3632892}
\showDOI{\tempurl}


\bibitem[Rajani(2020)]%
        {rajani>thesis}
\bibfield{author}{\bibinfo{person}{Vineet Rajani}.} \bibinfo{year}{2020}\natexlab{}.
\newblock \emph{\bibinfo{title}{A Type-Theory for Higher-Order Amortized Analysis}}.
\newblock {{doctoralThesis}}. \bibinfo{school}{Saarl\"andische Universit\"ats- und Landesbibliothek}.
\newblock
\urldef\tempurl%
\url{https://doi.org/10.22028/D291-30877}
\showDOI{\tempurl}


\bibitem[Rajani et~al\mbox{.}(2024)]%
        {rajani-barthe-garg>2024}
\bibfield{author}{\bibinfo{person}{Vineet Rajani}, \bibinfo{person}{Gilles Barthe}, {and} \bibinfo{person}{Deepak Garg}.} \bibinfo{year}{2024}\natexlab{}.
\newblock \showarticletitle{A {{Modal Type Theory}} of {{Expected Cost}} in {{Higher-Order Probabilistic Programs}}}.
\newblock \bibinfo{journal}{\emph{Proc. ACM Program. Lang.}} \bibinfo{volume}{8}, \bibinfo{number}{OOPSLA2} (\bibinfo{date}{Oct.} \bibinfo{year}{2024}), \bibinfo{pages}{285:389--285:414}.
\newblock
\urldef\tempurl%
\url{https://doi.org/10.1145/3689725}
\showDOI{\tempurl}


\bibitem[Rajani et~al\mbox{.}(2021)]%
        {rajani-gaboardi-garg-hoffmann>2021}
\bibfield{author}{\bibinfo{person}{Vineet Rajani}, \bibinfo{person}{Marco Gaboardi}, \bibinfo{person}{Deepak Garg}, {and} \bibinfo{person}{Jan Hoffmann}.} \bibinfo{year}{2021}\natexlab{}.
\newblock \showarticletitle{A Unifying Type-Theory for Higher-Order (Amortized) Cost Analysis}.
\newblock \bibinfo{journal}{\emph{Proceedings of the ACM on Programming Languages}} \bibinfo{volume}{5}, \bibinfo{number}{POPL} (\bibinfo{date}{Jan.} \bibinfo{year}{2021}), \bibinfo{pages}{27:1--27:28}.
\newblock
\urldef\tempurl%
\url{https://doi.org/10.1145/3434308}
\showDOI{\tempurl}


\bibitem[Reynolds(1983)]%
        {reynolds>1983}
\bibfield{author}{\bibinfo{person}{John~C. Reynolds}.} \bibinfo{year}{1983}\natexlab{}.
\newblock \showarticletitle{Types, {{Abstraction}}, and {{Parametric Polymorphism}}}. In \bibinfo{booktitle}{\emph{Information Processing 83, Proceedings of the {{IFIP}} 9th World Computer Congress, Paris, France, September 19-23, 1983}}, \bibfield{editor}{\bibinfo{person}{R.~E.~A. Mason}} (Ed.). \bibinfo{publisher}{North-Holland/IFIP}, \bibinfo{pages}{513--523}.
\newblock
\urldef\tempurl%
\url{https://doi.org/10.1007/3-540-55511-0_1}
\showDOI{\tempurl}


\bibitem[Riehl and Shulman(2017)]%
        {riehl-shulman>2017}
\bibfield{author}{\bibinfo{person}{Emily Riehl} {and} \bibinfo{person}{Michael Shulman}.} \bibinfo{year}{2017}\natexlab{}.
\newblock \showarticletitle{A Type Theory for Synthetic {$\infty$}-Categories}.
\newblock \bibinfo{journal}{\emph{Higher Structures}} \bibinfo{volume}{1}, \bibinfo{number}{1} (\bibinfo{date}{Dec.} \bibinfo{year}{2017}), \bibinfo{pages}{147--224}.
\newblock
\urldef\tempurl%
\url{https://doi.org/10.21136/HS.2017.06}
\showDOI{\tempurl}


\bibitem[Rijke et~al\mbox{.}(2020)]%
        {rijke-shulman-spitters>2020}
\bibfield{author}{\bibinfo{person}{Egbert Rijke}, \bibinfo{person}{Michael Shulman}, {and} \bibinfo{person}{Bas Spitters}.} \bibinfo{year}{2020}\natexlab{}.
\newblock \showarticletitle{Modalities in Homotopy Type Theory}.
\newblock \bibinfo{journal}{\emph{Logical Methods in Computer Science}}  \bibinfo{volume}{Volume 16, Issue 1} (\bibinfo{date}{Jan.} \bibinfo{year}{2020}).
\newblock
\urldef\tempurl%
\url{https://doi.org/10.23638/LMCS-16(1:2)2020}
\showDOI{\tempurl}


\bibitem[Sleator and Tarjan(1985a)]%
        {sleator-tarjan>1985-paging}
\bibfield{author}{\bibinfo{person}{Daniel~D. Sleator} {and} \bibinfo{person}{Robert~E. Tarjan}.} \bibinfo{year}{1985}\natexlab{a}.
\newblock \showarticletitle{Amortized Efficiency of List Update and Paging Rules}.
\newblock \bibinfo{journal}{\emph{Commun. ACM}} \bibinfo{volume}{28}, \bibinfo{number}{2} (\bibinfo{date}{Feb.} \bibinfo{year}{1985}), \bibinfo{pages}{202--208}.
\newblock
\showISSN{0001-0782}
\urldef\tempurl%
\url{https://doi.org/10.1145/2786.2793}
\showDOI{\tempurl}


\bibitem[Sleator and Tarjan(1985b)]%
        {sleator-tarjan>1985}
\bibfield{author}{\bibinfo{person}{Daniel~Dominic Sleator} {and} \bibinfo{person}{Robert~Endre Tarjan}.} \bibinfo{year}{1985}\natexlab{b}.
\newblock \showarticletitle{Self-Adjusting Binary Search Trees}.
\newblock \bibinfo{journal}{\emph{J. ACM}} \bibinfo{volume}{32}, \bibinfo{number}{3} (\bibinfo{date}{July} \bibinfo{year}{1985}), \bibinfo{pages}{652--686}.
\newblock
\showISSN{0004-5411}
\urldef\tempurl%
\url{https://doi.org/10.1145/3828.3835}
\showDOI{\tempurl}


\bibitem[Sterling and Harper(2021)]%
        {sterling-harper>2021}
\bibfield{author}{\bibinfo{person}{Jonathan Sterling} {and} \bibinfo{person}{Robert Harper}.} \bibinfo{year}{2021}\natexlab{}.
\newblock \showarticletitle{Logical {{Relations}} as {{Types}}: {{Proof-Relevant Parametricity}} for {{Program Modules}}}.
\newblock \bibinfo{journal}{\emph{J. ACM}} \bibinfo{volume}{68}, \bibinfo{number}{6} (\bibinfo{date}{Oct.} \bibinfo{year}{2021}), \bibinfo{pages}{41:1--41:47}.
\newblock
\showISSN{0004-5411}
\urldef\tempurl%
\url{https://doi.org/10.1145/3474834}
\showDOI{\tempurl}


\bibitem[Street(1974)]%
        {street>1974}
\bibfield{author}{\bibinfo{person}{Ross Street}.} \bibinfo{year}{1974}\natexlab{}.
\newblock \showarticletitle{Fibrations and {{Yoneda}}'s Lemma in a 2-Category}. In \bibinfo{booktitle}{\emph{Category {{Seminar}}}}, \bibfield{editor}{\bibinfo{person}{Gregory~M. Kelly}} (Ed.). \bibinfo{publisher}{Springer}, \bibinfo{address}{Berlin, Heidelberg}, \bibinfo{pages}{104--133}.
\newblock
\showISBNx{978-3-540-37270-7}
\urldef\tempurl%
\url{https://doi.org/10.1007/BFb0063102}
\showDOI{\tempurl}


\bibitem[Tarjan(1985)]%
        {tarjan>1985}
\bibfield{author}{\bibinfo{person}{Robert~Endre Tarjan}.} \bibinfo{year}{1985}\natexlab{}.
\newblock \showarticletitle{Amortized {{Computational Complexity}}}.
\newblock \bibinfo{journal}{\emph{SIAM Journal on Algebraic Discrete Methods}} \bibinfo{volume}{6}, \bibinfo{number}{2} (\bibinfo{date}{April} \bibinfo{year}{1985}), \bibinfo{pages}{306--318}.
\newblock
\showISSN{0196-5212}
\urldef\tempurl%
\url{https://doi.org/10.1137/0606031}
\showDOI{\tempurl}


\bibitem[{The Univalent Foundations Program}(2013)]%
        {univalentfoundations>2013}
\bibfield{author}{\bibinfo{person}{{The Univalent Foundations Program}}.} \bibinfo{year}{2013}\natexlab{}.
\newblock \bibinfo{booktitle}{\emph{Homotopy {{Type Theory}}: {{Univalent Foundations}} of {{Mathematics}}}}.
\newblock \bibinfo{publisher}{Univalent Foundations Program}.
\newblock


\bibitem[V{\'a}k{\'a}r(2017)]%
        {vakar>thesis}
\bibfield{author}{\bibinfo{person}{Matthijs V{\'a}k{\'a}r}.} \bibinfo{year}{2017}\natexlab{}.
\newblock \emph{\bibinfo{title}{In {{Search}} of {{Effectful Dependent Types}}}}.
\newblock {{http://purl.org/dc/dcmitype/Text}}. \bibinfo{school}{University of Oxford}.
\newblock
\urldef\tempurl%
\url{https://ora.ox.ac.uk/objects/uuid:e91e19b3-7e10-4fda-9433-f23b469e4049}
\showURL{%
\tempurl}


\bibitem[{van Br{\"u}gge}(2024)]%
        {vanbrugge>2024}
\bibfield{author}{\bibinfo{person}{Jan {van Br{\"u}gge}}.} \bibinfo{year}{2024}\natexlab{}.
\newblock \showarticletitle{Liquid {{Amortization}}: {{Proving Amortized Complexity}} with {{LiquidHaskell}} ({{Functional Pearl}})}. In \bibinfo{booktitle}{\emph{Proceedings of the 17th {{ACM SIGPLAN International Haskell Symposium}}}} \emph{(\bibinfo{series}{Haskell 2024})}. \bibinfo{publisher}{Association for Computing Machinery}, \bibinfo{address}{New York, NY, USA}, \bibinfo{pages}{97--108}.
\newblock
\showISBNx{979-8-4007-1102-2}
\urldef\tempurl%
\url{https://doi.org/10.1145/3677999.3678282}
\showDOI{\tempurl}


\bibitem[Vezzosi et~al\mbox{.}(2019)]%
        {vezzosi-mortberg-abel>2019}
\bibfield{author}{\bibinfo{person}{Andrea Vezzosi}, \bibinfo{person}{Anders M{\"o}rtberg}, {and} \bibinfo{person}{Andreas Abel}.} \bibinfo{year}{2019}\natexlab{}.
\newblock \showarticletitle{Cubical {{Agda}}: {{A Dependently Typed Programming Language}} with {{Univalence}} and {{Higher Inductive Types}}}.
\newblock \bibinfo{journal}{\emph{Proceedings of the ACM on Programming Languages}} \bibinfo{volume}{3}, \bibinfo{number}{ICFP} (\bibinfo{date}{July} \bibinfo{year}{2019}), \bibinfo{pages}{87:1--87:29}.
\newblock
\urldef\tempurl%
\url{https://doi.org/10.1145/3341691}
\showDOI{\tempurl}


\bibitem[Xu and Wang(2026)]%
        {xu-wang>2026}
\bibfield{author}{\bibinfo{person}{Han Xu} {and} \bibinfo{person}{Di Wang}.} \bibinfo{year}{2026}\natexlab{}.
\newblock \showarticletitle{Dependently-{{Typed AARA}}: {{A Non-Affine Approach}} for {{Resource Analysis}} of {{Higher-Order Programs}}}. In \bibinfo{booktitle}{\emph{Programming {{Languages}} and {{Systems}}}}, \bibfield{editor}{\bibinfo{person}{Robbert Krebbers}} (Ed.). \bibinfo{publisher}{Springer Nature Switzerland}, \bibinfo{address}{Cham}, \bibinfo{pages}{362--391}.
\newblock
\showISBNx{978-3-032-22723-2}
\urldef\tempurl%
\url{https://doi.org/10.1007/978-3-032-22723-2_13}
\showDOI{\tempurl}


\end{thebibliography}

\end{document}